\documentclass[a4p,12pt]{article}

\usepackage{amsfonts,amsmath,amssymb,amsthm,ascmac,mathrsfs,mathtools,comment,enumerate,bm,multicol,setspace,lmodern,microtype,physics,paralist}
\usepackage[super]{nth}

\usepackage{graphicx}
\usepackage{hyperref}
\usepackage{blindtext}
\usepackage[dvipsnames]{xcolor}
\allowdisplaybreaks

\hypersetup{colorlinks=true, linkcolor=black, citecolor=black, urlcolor=Blue}

\usepackage{natbib}
\setcitestyle{authoryear,open={(},close={)}}

\newtheorem{theorem}{Theorem}
\newtheorem{proposition}{Proposition}
\newtheorem*{corollary}{Corollary}
\newtheorem{lemma}{Lemma}

\newtheorem*{example}{Example}

\theoremstyle{definition}
\newtheorem{definition}{Definition}
\newtheorem{assumption}{Assumption}

\theoremstyle{remark}
\newtheorem{remark}{Remark}

\usepackage[right=1in, left=1in, bottom=1in, top=1in]{geometry}

\newcommand{\N}{\mathbb{N}}
\newcommand{\R}{\mathbb{R}}

\newcommand{\1}{\mathbf{1}}

\renewcommand{\d}{{\rm d}}

\newcommand{\scrW}{\mathscr{W}}
\newcommand{\scrX}{\mathscr{X}}

\newcommand{\calF}{\mathcal{F}}
\newcommand{\calG}{\mathcal{G}}

\newcommand{\calL}{\mathcal{L}}

\newcommand{\calN}{\mathcal{N}}

\newcommand{\calP}{\mathcal{P}}

\newcommand{\calS}{\mathcal{S}}

\newcommand{\calX}{\mathcal{X}}

\newcommand{\argmax}{\mathop{\rm arg~max}\limits}

\newcommand{\vertd}[1]{{\left\vert\kern-0.25ex\left\vert\kern-0.25ex\left\vert #1 \right\vert\kern-0.25ex\right\vert\kern-0.25ex\right\vert}}
\newcommand{\vertb}[1]{{\bigl\vert\kern-0.25ex\bigl\vert\kern-0.25ex\bigl\vert #1 \bigr\vert\kern-0.25ex\bigr\vert\kern-0.25ex\bigr\vert}}
\newcommand{\vertt}[1]{{\vert\kern-0.25ex\vert\kern-0.25ex\vert #1 \vert\kern-0.25ex\vert\kern-0.25ex\vert}}

\DeclareMathOperator{\E}{\mathbb{E}}
\DeclareMathOperator{\Var}{\mathbb{V}ar}
\DeclareMathOperator{\Cov}{\mathbb{C}ov}
\DeclareMathOperator{\Corr}{\mathbb{C}orr}

\renewcommand{\P}{\mathbb{P}}

\DeclareMathOperator{\Diag}{Diag}
\DeclareMathOperator{\sgn}{sgn}

\begin{document}

\title{{\bf Identification of Information Structures in Bayesian Games}\thanks{This paper is based on the first chapter of my dissertation submitted to Yale University.
I am deeply indebted to my advisors, Mira Frick, Johannes H\"orner, Ryota Iijima, and Larry Samuelson, for their invaluable advice and encouragement throughout the project.
I benefited greatly from constructive comments from an associate editor and four anonymous referees through the review process.
For helpful comments and discussions, I am grateful to Dirk Bergemann, Krishna Dasaratha, Tetsuya Hoshino, Jonathan Libgober, Hitoshi Matsushima, Kyohei Okumura, Aniko \"Ory, Philipp Strack, Takashi Ui, Katsutoshi Wakai, and Kai Hao Yang, as well as seminar and conference participants at Hitotsubashi, HKU, ITAM, Kyoto, Yale, SWET2023, and the 33rd Stony Brook Conference.}
}
\author{Masaki Miyashita\thanks{The University of Hong Kong; {\tt masaki11@hku.hk}
}}
\date{\today}
\maketitle

\bigskip

\renewcommand{\baselinestretch}{1.1}
\begin{abstract}
To what extent can an external observer infer the underlying information structure from an equilibrium action distribution in an incomplete-information game?
We investigate this question in a general linear-quadratic-Gaussian framework.
A simple class of canonical information structures is offered and proves rich enough to rationalize any equilibrium action distribution that can arise under an arbitrary information structure.
Moreover, this class is parsimonious in the sense that its relevant parameters are uniquely identified from an observed equilibrium outcome.
We then show that a canonical information structure characterizes the lower bound on the amount by which each agent's signal can reduce the state variance, across all observationally equivalent information structures.
This identified lower bound can in turn be used to predict equilibrium action volatility following changes in the payoff structure.
\end{abstract}

\renewcommand{\baselinestretch}{1}
\hypersetup{colorlinks=true, linkcolor=Red, citecolor=Blue}

\newpage
\renewcommand{\baselinestretch}{1}
\onehalfspacing
\section{Introduction}

Imagine a situation where a population of agents is engaged in an incomplete-information game.
There is an outside observer---called an \emph{econometrician}---who possesses only partial knowledge of the game structure.
This paper asks to what extent the econometrician can learn about agents' informational traits, such as the informativeness of each agent's signal about the state or the correlation between different agents' signals, by observing the equilibrium outcome of the game.
In other words, our interest lies in the identification of the underlying information structure.

Knowing agents' informational traits is policy relevant because of their decisive role in shaping equilibrium behavior.
For example, in a standard model of market competition, firms choose production levels based on their private information about an unknown demand function.
At the same time, their strategic responses to policy changes that alter the payoff structure of the game depend on the underlying information structure.
For instance, the introduction of an excise tax reduces firms' production levels, but private information affects how much production shrinks in equilibrium.
Thus, a government may benefit from learning the underlying information structure in order to make reliable predictions about counterfactual economic outcomes following policy changes.\footnote{In this paper, ``counterfactual'' prediction refers to prediction of economic outcomes following a potential change in the payoff structure, rather than, for example, prediction in the identical game after redrawing relevant random variables such as the state and signals.}
The fundamental difficulty is that the information structure is not directly observable, but rather need be inferred from observable data on firms' past choices.

The presence of strategic interactions poses a crucial obstacle to identification by confounding agents' private information about fundamental and strategic uncertainty in equilibrium behavior, in proportions unknown to the econometrician.
To illustrate this, suppose that a fairly high correlation is observed between the state realization and an agent's action.
On the one hand, this correlation may reflect the high predictive accuracy of the agent's private signal about the state.
On the other hand, it may be attributed to the agent's payoff-relevant motive to adjust her action to those of other agents who are better informed about the state.
Hence, the observed correlation may result from a mixture of the direct correlation between the state and the agent's own signal and the spurious correlation generated through strategic interactions with others.

Our analysis is conducted in a general framework of \emph{linear-quadratic-Gaussian} (LQG) games \`a la \cite{radner1962} with a continuum of agents, which maintains two structural assumptions.
First, each agent's payoff function is quadratic, so that the induced best-response strategy is linear in her best estimates about the state and other agents' actions.
Second, the state and signals jointly follow a Gaussian process, which is an infinite-dimensional analogue of a normal distribution.
Apart from these requirements, our framework allows for general payoff structures, where actions may be strategic complements or substitutes, as well as general correlation structures among agents' signals and the state.

The focus on LQG games allows us to sidestep issues of equilibrium non-existence or multiplicity.
Theorem~\ref{thm_well} shows that, for almost every LQG game, there exists a unique affine equilibrium, that is, an equilibrium in which each agent responds linearly to her private signal, and this equilibrium varies continuously with the primitives.
This result holds without imposing substantial parameter restrictions, such as homogeneity across agents or restrictions on the sign or magnitude of strategic interactions.

Our identification analysis proceeds in steps.
First, we offer a simple class of \emph{canonical information structures}, under which each agent receives a one-dimensional signal that is additively decomposed into the state-relevant term and idiosyncratic noise.
Theorem~\ref{thm_equiv} shows that this class is rich enough to rationalize any equilibrium outcome that can arise under an arbitrary information structure.
From a modeling perspective, therefore, the modeler incurs no loss of generality by assuming as if agents' signals are drawn from some  canonical information structure.
On the other hand, the theorem points to the limit of full identification, since it implies that any information structure is observationally equivalent to some canonical information structure.

Second, we show that, by restricting attention to the canonical class, the econometrician can fully identify the relevant parameters of an information structure.
Specifically, any canonical information structure is parameterized by two functions, one representing the exposure of each agent's signal to the state and the other governing correlations among noise terms.
Theorem~\ref{thm_id} shows that these functions are uniquely determined from observables, under a minimal condition on action variability.
Our identification approach takes the form of a reduced-form regression.
The canonical representation allows us to separate the variation in each agent's action attributed to changes in the state from the residual variation attributed to orthogonal noise, and the parameters of interest are recovered from the corresponding coefficients.
A key feature of this reduced-form approach is that these objects are recovered directly from the observed action distribution, and thus Theorem~\ref{thm_id} requires no prior knowledge of the underlying payoff structure.

Third, we study how the identified canonical information structure is related to a general information structure that is observationally equivalent to it.
This issue is important from a counterfactual viewpoint because observational equivalence between two information structures depends on a given payoff structure.
For example, two information structures that generate the same market outcome under an initial tax rate may no longer do so after a change in the tax rate.
Thus, eliciting a meaningful relationship between these information structures is useful for making reliable counterfactual predictions.

Our starting point is a rather simple observation that an agent's general signal can contain more information than the agent's action, while the canonical signal serves, by construction, as an action recommendation in the spirit of Bayes correlated equilibrium \citep{bm2013}.
Put differently, the canonical signal contains the minimal amount of information needed for the agent to play the same action.
Building on this observation, Proposition~\ref{prop_var} shows that a canonical information structure characterizes a lower bound on the informativeness of each agent's signal, as measured by state variance reduction, across all information structures that can be observationally equivalent to it under some payoff structure.
Moreover, Proposition~\ref{prop_var_eq} derives an equality condition for this bound, which is satisfied, e.g., when there are no strategic interactions among agents.
In general, however, strategic interactions create a gap between the lower bound and the actual variance reduction, reflecting the identification difficulty posed by the aforementioned confounding issue.

Lastly, we showcase how our results can be translated into economic predictions.
For this exercise, we focus on aggregate action volatility, an equilibrium statistic that has direct welfare implications in a wide class of LQG games \citep{ms2002,ap2007,ui2015}.
Proposition~\ref{prop_actvar} provides a lower bound on aggregate action volatility in terms of aggregate state variance reduction and payoff parameters.
Combining this result with the previous identification results then yields a robust counterfactual prediction of aggregate action volatility following an exogenous change in payoff parameters.
In the context of our running example, this implies that the econometrician can identify guaranteed levels of producer surplus and expected tax revenue following a change in the tax rate.

The rest of the paper is organized as follows.
After discussing related literature, Section~\ref{sec_model} sets up the model of LQG games.
Section~\ref{sec_eq} studies the theoretical properties of equilibria.
Section~\ref{sec_id} conducts the main identification analysis.
Section~\ref{sec_conc} concludes.
The appendices collect omitted proofs and additional results.

\paragraph{Related literature.}
This paper borrows the framework of \cite{bm2013}, a seminal paper that studies how an analyst can make robust predictions about economic outcomes without knowing the information structure agents face.
A sequence of papers \citep{bm2013, bm2016,bhm2015,bhm2017,bhm2021} addresses this problem by characterizing the set of equilibrium outcomes that can arise under any information structure on the basis of Bayes correlated equilibrium (BCE).
\cite{lop2018} also provide robust predictions of market variables, such as prices, in the context of asset markets.
We put forward this line of analysis by allowing the analyst to access data on past equilibrium play and thereby sharpen her knowledge of the underlying information structure.
In addition, while the BCE-based approach fixes the payoff structure  and instead obtain a set prediction of equilibrium outcomes across information structures, our analysis addresses counterfactual predictions upon potential changes in payoff structures.
Our partial identification results speak to how past observations of equilibrium outcomes can narrow down information structures and yield counterfactual predictions following policy changes.

In comparison to the vast literature on belief elicitation in statistics \citep[e.g.,][]{savage1971} and economics \citep[e.g.,][]{karni2009}, this paper aims to identify the ex-ante joint distribution of the state and signals, rather than eliciting subjective beliefs about the state formed after signal realizations.
There are also contemporaneous papers that share a similar goal of identifying information structures, such as \cite{lu2016,lu2019}, \cite{arieli2017}, and \cite{libgober2025}.
In contrast with these papers, our analysis emphasizes strategic interactions among agents, which hinder the identification of agents' informational traits by confounding their private information through equilibrium strategies.

There is also a related brunch of econometrics literature on identification in linear interaction models \citep{manski1993,manski1995,lee2007,bramoulle2009}.
As \cite{bm2013} clarify, the regression models in these papers correspond to a reduced-form representation of Nash equilibrium in network games under complete information.
\cite{blume2015} build a structural model of linear interaction games and extend the identification analysis to incomplete-information settings.
While this strand of literature focuses on the identification of payoff parameters, the informational connections across agents constitute the parameters of interest for us.\footnote{\cite{bm2013} also address partial identification of payoff structures.
In this regard, this paper is complementary to their analysis, as we are interested in the identification of information structures.}

Apart from identification analysis, this paper's theoretical contribution lies in proving a general result on equilibrium well-posedness in continuum-population LQG games.
We develop a novel proof technique by leveraging the Riesz--Fredholm theory of integral equations, and show that unique existence of affine equilibrium holds for almost every LQG game, without imposing homogeneity across agents \citep{ms2002,ap2007,bm2013,ui2013}, restricting the sign and/or magnitude of strategic interactions \citep{graphon2023}, or assuming a finite population \citep{ui2016,lmo2018}.
We also address the well-known measurability problem \citep{judd1985, uhlig1996} concerning the aggregation of a continuum of random variables in continuum-population settings.
As a remedy for this issue, we advocate the use of the integral notion \`a la \cite{pettis1938} and showcase how the properties of the Pettis integral allow us to carry over much of the intuition from finite-population models to the continuum-population setting.
Finally, the continuity result in Theorem~\ref{thm_well} is an addition to the literature, which enables us to interpret the present model with  continuously many heterogenous agents as the limit of large finite-population analogues, in a spirit similarly to \cite{graphon2023}.


\section{Model} \label{sec_model}


\subsection{Payoff Structure} \label{sec_payoff}

There are a continuum of agents $i \in [0,1]$, each of whom simultaneously chooses an action $x(i) \in \R$.
Given an action profile $\{x(i)\}_{i \in [0,1]}$, the \emph{local aggregate} for agent $i$ is defined as
\[
y(i) = \int_0^1 w(i,j)x(j) \d j,
\]
where $w: [0,1]^2 \to \R$ is a weight function with $w(i,j)$ measuring the influence of agent $j$'s action on agent $i$'s local aggregate.
The function $w$ need not be binary-valued or symmetric, and it allows for heterogeneous signs and magnitudes of strategic interaction.

Each agent $i$'s utility function $u_i:\R^3 \to \R$ is quadratic in own action $x(i)$, local aggregate $y(i)$, and the payoff-relevant state $\theta\in\R$, and concave in $x(i)$. 
Concretely, we consider the following general quadratic form:
\begin{equation} \label{eq_q}
u_i \qty(x(i),y(i),\theta) = z(i)^\top U(i) z(i), \quad {\rm where} \quad z(i)^\top = \mqty[x(i) & y(i) & \theta & 1].
\end{equation}
with $U(i) = [U_{kl}(i)]_{4\times 4}$ a symmetric matrix of quadratic coefficients satisfying $U_{11}(i)<0$.
We define $a(i) \coloneqq -U_{12}(i)/U_{11}(i)$, $b(i) \coloneqq -U_{13}(i)/U_{11}(i)$, and $c(i) \coloneqq -U_{14}(i)/U_{11}(i)$.

The quadratic specification implies linear best responses.
Under complete information, given any action profile $\{x(j)\}_{j\neq i}$ of other agents, the first-order condition yields
\[
x(i) =  a(i)\int_0^1 w(i,j)x(j)\d j + b(i)\theta + c(i),
\]
and an analogous affine best-response formula applies under incomplete information. 
Thus $w$ together with the first row of $U$---i.e., the functions $a$, $b$, and $c$---suffice to describe agents' incentives. 
Without loss of generality, we normalize $a(i)=1$ for all $i\in[0,1]$ by absorbing heterogeneity in $a$ into the weight function $w$.
We therefore refer to the profile $(w,b,c)$ as the \emph{payoff structure}. 
Throughout, we assume $b$, $c$, and $w$ are square-integrable functions,\footnote{That is, $\int_0^1 |b(i)|^2\d i < \infty$, $\int_0^1 |c(i)|^2\d i < \infty$, and $\int_0^1\int_0^1 |w(i,j)|^2\d i\d j < \infty$.} but are otherwise unrestricted.

\begin{example}[Market Competition]
As a running example, we use a stylized model of market competition, 
which has been repeatedly studied in the LQG literature; see, e.g., \cite{vives1984,vives1999} and \cite{bm2013}.
Agents $i \in [0,1]$ are interpreted as firms producing a homogeneous good. 
Each firm chooses a production level $x(i) \in \R$ at the quadratic cost $x(i)^2/2$. 
Market demand depends on aggregate production $y=\int_0^1 x(i)\,\d i$ and on an unknown demand intercept $\theta \in \R$. 
Each firm's profit is
\[
u(x(i),y,\theta) = p^\tau(y,\theta)\cdot x(i) - \frac{x(i)^2}{2},  \quad \text{where} \quad p^\tau(y,\theta)=(1-\tau)(\theta-y).
\]
Here, $\tau\in[0,1)$ denotes the excise tax rate set by the authority, so that firms face the net price $p^\tau(y,\theta)$ per unit of output. 
This payoff function is quadratic in $(x(i),y,\theta)$, with the associated parameters $w(i,j) = -(1-\tau)$, $b(i) = 1-\tau$, and $c(i) = 0$ for all $i,j \in [0,1]$.
\end{example}


\subsection{Information Structure} \label{sec_info}

We fix the common prior over the state $\theta$ to be normal, $\theta \sim \calN(\mu_\theta,\sigma_\theta^2)$, with mean $\mu_\theta\in\R$ and variance $\sigma_\theta^2>0$. 
Before choosing an action, each agent $i$ observes a signal realization $s(i)\in\R^{\bar d}$ of a possibly multi-dimensional random vector $S(i)$, where $\bar d\in\N$. 
We specify the joint distribution of the state and signals as a \emph{Gaussian process}, meaning that any finite subset of the collection
\[
\calS \coloneqq \{\theta\} \cup \{S_1(i),\ldots,S_{\bar d}(i)\}_{i\in[0,1]}
\]
is multivariate normal.

A key advantage of Gaussian modeling is that the signal-state joint distribution is summarized by its first and second moments. 
Let $S(i)=[S_1(i),\ldots,S_{\bar d}(i)]^\top$ be the column vector of agent $i$'s signal, with mean $m(i)\coloneqq\E[S(i)]\in\R^{\bar d}$ and covariance matrices given by
\[
K_\theta(i)\coloneqq\Cov\qty[S(i),\theta]\in\R^{\bar d}, \qquad
K(i,j)\coloneqq\Cov\qty[S(i),S(j)]\in\R^{\bar d\times\bar d}.
\]
Moreover, for any finite set $N=\{i_1,\ldots,i_n\}\subseteq[0,1]$, we write the joint covariance matrix as
\begin{equation} \label{cov_matrix_join}
\mqty[{\bm K}(N) & {\bm K}_\theta(N) \\ {\bm K}_\theta(N)^\top & \sigma^2_\theta] \coloneqq
\mqty[K(i_1,i_1) & \cdots & K(i_1,i_n) & K_\theta(i_1) \\ \vdots & \ddots & \vdots & \vdots \\ K(i_n,i_1) & \cdots & K(i_n,i_n) & K_\theta(i_n) \\ K_\theta(i_1)^\top & \cdots & K_\theta(i_n)^\top & \sigma^2_\theta]
\in \R^{(\bar{d}n+1) \times (\bar{d}n+1)},
\end{equation}
which is symmetric and positive semidefinite by construction.
These two properties are also sufficient for the existence of the corresponding Gaussian process.\footnote{According to Theorem 12.1.3 of \cite{dudley}, given any set $T$ and any function $\mu: T\to \R$, if a function $\sigma:T^2 \to \R$ satisfies symmetry and positive semidefiniteness, then there exists a probability space on which a Gaussian process over $T$ having the mean function $\mu$ and the covariance function $\sigma$ takes place.}
Without loss of generality, assume that $K(i,i)$ is invertible for every $i \in [0,1]$.

Since $K(i,i)$ is positive definite, its inverse admits a unique symmetric square root $K(i,i)^{-\frac{1}{2}}$. 
Using it, we define the standardized covariance matrices
\[
P(i,j)\coloneqq K(i,i)^{-\frac{1}{2}}K(i,j)K(j,j)^{-\frac{1}{2}}, \qquad
P_\theta(i)\coloneqq \frac{1}{\sigma_\theta}\,K(i,i)^{-\frac{1}{2}}K_\theta(i),
\]
which are matrix-valued analogues of correlation coefficients.
By construction, $P(i,i)$ is the identity matrix.
As we will see, equilibrium behavior depends on an information structure only through the functions $P:[0,1]^2\to\R^{\bar d\times\bar d}$ and $P_\theta:[0,1]\to\R^{\bar d}$.
We therefore refer to $\calP=(P,P_\theta)$ as an \emph{information structure}.
Throughout, assume that each coordinate of $P$ and $P_\theta$ is measurable as a function of $i \in [0,1]$.\footnote{Unlike the parameters of the payoff structure, we do not need to assume square integrability of $P$ and $P_\theta$ explicitly because their values are bounded by construction. This boundedness resembles that of correlation coefficients and is established formally in Lemma~\ref{lem_P}. Also, no condition on $m$ is required, since only unbiased signals matter for agents' strategies \eqref{eq_a}.}


\subsection{Strategy and Equilibrium} \label{sec_def}

An \emph{LQG game} $\calG = (w,b,c,\calP)$ consists of a payoff structure $(w,b,c)$ and an information structure $\calP = (P,P_\theta)$. 
The profile $\calG$ is assumed to be common knowledge among all agents, while an external observer may know it only partially or not at all. 
The observer's goal is to recover the information structure $\calP$ from an observed equilibrium outcome.

\begin{definition} \label{def_st}
Each agent $i$'s \emph{strategy} is a measurable mapping $X(i): \R^{\bar d} \to \R$ such that $\E|X(i)|^2 < \infty$. 
A collection of strategies $X(\cdot) \equiv \{X(i)\}_{i \in [0,1]}$ is a \emph{strategy profile} if it satisfies the following regularity conditions:
\begin{inparaenum}[(i)]
\item the mapping $i \mapsto \E[X(i)X(j)]$ is measurable for each $j \in [0,1]$; and
\item the mapping $i \mapsto \E|X(i)|^2$ is integrable.
\end{inparaenum}
\end{definition}

\begin{remark} \label{remark_int}
In a continuum-population LQG game, a strategy profile $X(\cdot) = \{X(i)\}_{i \in [0,1]}$ forms a stochastic process on $[0,1]$, with randomness in actions stemming from that in signals. 
Thus, we must integrate such processes to compute stochastic local aggregates of the form $Y(i)=\int_0^1 w(i,j) X(j),\mathrm{d}j$.
A technical difficulty arises because, when each $X(i)$ contains a purely idiosyncratic component independent of others, a typical sample path of the process is highly discontinuous and fails to satisfy the needed measurability for path-wise integration; see, e.g., \citet{judd1985} and \cite{uhlig1996}.
To avoid this measurability issue, we adopt the integral notion \`a la \cite{pettis1938} to define the integration of strategy profiles.\footnote{The use of Pettis integral is not new in economics, but rather, it also appear in earlier works such as \cite{alnajjar1995} and \cite{uhlig1996}; see also Chapter~11 of \cite{ab2006}.}
Appendix~\ref{app_pettis} provides the formal definition of the Pettis integral and confirms that the regularity conditions in Definition~\ref{def_st} ensure the Pettis integrability of any strategy profile.
\end{remark}

\begin{definition}
A strategy profile $X(\cdot)$ forms a \emph{Bayesian Nash equilibrium} if for any $i \in [0,1]$ and $s(i) \in \R^{\bar d}$,
\[
X(i)[s(i)] \in \argmax_{x \in \R} \E \qty[u_i(x, Y(i), \theta) \mid S(i) = s(i)],
\]
where $Y(i) = \int_0^1 w(i,j) X(j)\d j$ is defined as the Pettis integral.
We write $E_i[\cdot] = \E[\cdot \mid S(i)]$ for the expectation conditional on agent $i$'s signal.
\end{definition}

We focus on Bayesian Nash equilibria in which each agent's strategy depends linearly on their standardized signal. 
This standardization is without loss of generality, but it helps economize notation.

\begin{definition} \label{def_affine}
A strategy $X(i)$ is \emph{affine} if there exist $\varphi_0(i) \in \R$ and $\varphi(i) \in \R^{\bar d}$ such that
\begin{equation} \label{eq_a}
X(i)[s(i)] = \varphi_0(i) + \varphi(i)^\top K(i,i)^{-\frac{1}{2}}\qty(s(i)-m(i)), 
\quad \forall\, s(i) \in \R^{\bar d}.
\end{equation}
\end{definition}

Thus, $\varphi_0(i)$ captures the intercept of agent $i$'s strategy, while $\varphi(i)$ governs its sensitivity to the standardized signal. 
Since these components fully determine an affine strategy, a profile of affine strategies is identified as a collection of functions $\varphi_0:[0,1]\to\R$ and $\varphi=(\varphi_1,\ldots,\varphi_{\bar d}):[0,1]\to\R^{\bar d}$. 
We denote the full profile by $\bar{\varphi}=(\varphi_0,\varphi)$. 
When the regularity conditions in Definition~\ref{def_st} are satisfied, we refer to $\bar{\varphi}$ as an \emph{affine strategy profile}. 
If $\bar{\varphi}$ constitutes a Bayesian Nash equilibrium, it is called an \emph{affine equilibrium}.

\begin{definition}
For an affine strategy profile $\bar{\varphi}$ in an LQG game $\calG$, the collection of random variables 
\[
\calX \coloneqq \{\theta\} \cup \{X(i)\}_{i\in[0,1]}
\]
is called the \emph{induced outcome} of $(\calG,\bar{\varphi})$.
\end{definition}

An induced outcome will serve as a key building block in our identification exercise, where the econometrician aims to recover the primitive of an LQG game from the distribution of $\calX$. 
Since affine transformations preserve normality, any affine strategy profile induces a Gaussian distribution over the state and actions.
The next lemma characterizes the corresponding moments and also provides the necessary and sufficient condition for an affine strategy profile $\bar{\varphi}$ to be ``valid'' in the sense of Definition~\ref{def_st}. 
The proof is straightforward and omitted.

\begin{lemma} \label{lem_affine}
An affine strategy profile $\bar{\varphi}:[0,1] \to \R^{\bar d+1}$ satisfies the regularity conditions in Definition~\ref{def_st} if and only if $\varphi_d \in \calL_2[0,1]$ for all $d=0,1,\ldots,\bar d$, where $\calL_2[0,1]$ is the space of square-integrable functions on $[0,1]$. 
For any such $\bar{\varphi}$, the induced outcome is a Gaussian process with mean and covariance given by
\[
\E[X(i)] = \varphi_0(i), \quad
\Cov\qty[X(i),X(j)] = \varphi(i)^\top P(i,j)\varphi(j), \quad
\Cov\qty[X(i),\theta] = \sigma_\theta \varphi(i)^\top P_\theta(i).
\]
\end{lemma}

Thanks to the signal standardization in Definition~\ref{def_affine}, this lemma shows that $\varphi_0(i)$ coincides with the mean of $X(i)$, and $\varphi(i)$ captures the variability of $X(i)$ as $\Var\qty[X(i)] = \varphi(i)^\top \varphi(i)$ holds by the fact that $P(i,i)$ is the identity matrix.

\begin{example}[Continued]
The market competition model constitutes an LQG game when the unknown state $\theta$ and each firm's private signal $S(i)$ are jointly Gaussian distributed. 
Each firm's best-response strategy is linear in its conditional expectations of $\theta$ and of the random aggregate output $Y=\int_0^1 X(j)\d j$, as follows:
\[
X(i) = (1-\tau)\E_i\qty[\theta - Y].
\]
The equilibrium is shaped not only by the payoff parameters but also by the information structure through the conditional distributions of $Y$ and $\theta$.

In this model, producer surplus (PS) is defined as the expected aggregate profit of all firms.
Together with the law of iterated expectations, the above best-response condition implies that PS is calculated as
\begin{align*}
{\rm PS} &= \E\qty[\int_0^1 \qty(p^\tau(Y,\theta) \cdot X(i) - \frac{X(i)^2}{2}) \d i] \\
&= \int_0^1 \E \qty[(1-\tau)\E_i\qty[\theta - Y] \cdot X(i) - \frac{X(i)^2}{2}] \d i = \frac{1}{2} \int_0^1 \E\qty|X(i)|^2 \d i.
\end{align*}
This feature, whereby expected payoffs can be written in terms of equilibrium action volatility, is a common feature of LQG games; see, e.g., \cite{ui2015}.
\end{example}


\section{Equilibrium Analysis} \label{sec_eq}

In this section, we establish the unique existence of an affine equilibrium in an LQG game.
The analysis proceeds in two steps.
We first express agents' best-response conditions as a Fredholm integral equation.
We then show that, generically, this equation admits a unique solution, and the solution depends continuously on primitives.


\subsection{Generic Well-posedness of Equilibrium}

Let $X(\cdot)$ be a candidate strategy profile that is affine and represented by coefficient functions 
$\bar{\varphi} = (\varphi_0, \varphi)$.
For this profile to constitute an affine equilibrium, each agent $i$ must satisfy the first-order condition
\begin{equation} \label{eq_l}
X(i) = \E_i \qty[Y(i)] + b(i)\E_i\qty[\theta] + c(i).
\end{equation}
Substituting the definition \eqref{eq_a} of affine strategies into \eqref{eq_l} yields
\begin{equation} \label{eq_l2}
X(i) = \int_0^1 w(i,j)\qty( \varphi_0(j) + \varphi(j)^\top K(j,j)^{-\frac{1}{2}} \qty(\E_i[S(j)] - m(j))) \d j + b(i)\E_i[\theta] + c(i),
\end{equation}
where the interchange of expectation and integration is justified by the definition of the Pettis integral; see Appendix~\ref{app_pettis}.
Moreover, because $\calS$ is a Gaussian process, the conditional expectations $\E_i\qty[S(j)]$ and $\E_i\qty[\theta]$ are given as affine functions of $S(i)$,
\begin{gather*}
\E_i\qty[S(j)] = m(j) + K(j,i)K(i,i)^{-1}\qty(S(i) - m(i)), \\
\E_i\qty[\theta] = \mu_\theta + K_\theta(i)K(i,i)^{-1}\qty(S(i) - m(i)).
\end{gather*}
Substituting these expressions back into \eqref{eq_l2} yields
\begin{align} \label{eq_match}
X(i)
=&\qty(\int_0^1 w(i,j) \varphi_0(j) \d j+\mu_\theta b(i) + c(i)) \nonumber\\
&+ \qty(\int_0^1 w(i,j) P(i,j) \varphi(j) \d j + \sigma_\theta b(i)P_\theta(i))^\top
K(i,i)^{-\frac{1}{2}}\qty(S(i)-m(i)).
\end{align}
Then matching coefficients between \eqref{eq_a} and \eqref{eq_match} yields the following $(\bar d+1)$-dimensional integral equation characterizing an affine equilibrium:
\begin{equation} \tag{BNE} \label{int_bne}
\mqty[\varphi_0(i) \\ \varphi(i)] = \mqty[\displaystyle \int_0^1 w(i,j)\varphi_0(j) \d j \\ \displaystyle \int_0^1 w(i,j) P(i,j) \varphi(j) \d j] + \mqty[\mu_\theta b(i) + c(i) \\ \sigma_\theta b(i) P_\theta(i)], \quad \forall i \in [0,1].
\end{equation}
By Lemma~\ref{lem_affine}, the profile $\bar{\varphi}$ constitutes an affine equilibrium if and only if it satisfies \eqref{int_bne} and each component $\varphi_d$ belongs to $\calL_2[0,1]$ for all $d=0,1,\ldots,\bar d$.

Mathematically, the equation~\eqref{int_bne} is an instance of Fredholm integral equations (of the second kind) in the unknown function $\bar{\varphi}$, whose solvability is well studied in the classical Riesz--Fredholm theory of integral equations.
The next theorem establishes that, generically, this equation admits a unique solution that depends continuously on the primitives of the game.
Formally, we topologize the space of LQG games by the $L_2$-norm on their primitives and define genericity in terms of openness and denseness with respect to this norm.
The detailed functional-analytic setup and the proof are given in Appendix~\ref{app_pf}.

\begin{theorem} \label{thm_well}
There exists a unique affine equilibrium for almost every LQG game.
Moreover, if a sequence of LQG games $\{\calG_n\}_{n=1}^\infty$ converges to $\calG$, and if $\calG$ has a unique affine equilibrium $\bar{\varphi}$, then $\calG_n$ has a unique affine equilibrium $\bar{\varphi}_n$ for all large enough $n$, and $\bar{\varphi}_n$ converges to $\bar{\varphi}$.
\end{theorem}

This theorem generalizes the standard unique-existence results obtained for a continuum of homogeneous agents \citep{ap2007, bm2013} to the case with heterogeneous agents.
Notably, it imposes no parametric restrictions on the payoff or information structures, which makes the result substantially more general than existing ones.\footnote{Our model bears some similarity to the recent framework of continuum-population network games developed by \citet{graphon2023}, and Theorem~\ref{thm_well} is related to their unique-existence theorem.
Compared with their setting, the strength of Theorem~\ref{thm_well} lies first in incorporating incomplete information, whereas their model assumes a complete-information environment.
Moreover, while they adopt the common assumption in the network-game literature that the spectral radius of the integral kernel $w:[0,1]^2 \to \R$ is bounded, thereby limiting the strength of strategic interactions, our result holds for generic $w$ as well as for generic signal-state joint distributions.}
A related generic result is obtained by \citet{lmo2018} in a setting with finitely many agents.
Extending such results to the continuum-agent case requires careful treatment of measurability and integrability and is new in the literature.

A natural question is whether the uniqueness assertion of Theorem~\ref{thm_well} extends beyond affine strategies; that is, whether the affine equilibrium remains the only BNE even when non-affine strategies are considered.
\citet{ui2016} obtains a positive result in finite-population models by deriving a sufficient condition under which this strong uniqueness holds in the mean-square sense.
Although we do not fully pursue this kind of stronger result here and adopt the unique affine equilibrium in Theorem~\ref{thm_well} as the equilibrium prediction for the subsequent identification analysis, in Appendix~\ref{app_pettis}, we hint at how an analogous result would be obtained in the continuum-population setting.
Crucially, the definition of Pettis integral provides the technical foundation that allows the key argument for establishing mean-square uniqueness to extend to infinite populations, which makes it a compelling modeling choice from a technical standpoint.


\subsection{Equilibrium Moment Restrictions}

The equation~\eqref{int_bne} imposes restrictions on each agent's affine strategy by requiring it to equal the sum of the weighted average of others' strategies and constant terms.
Specifically, since $\varphi_0(i) = \E[X(i)]$ by Lemma~\ref{lem_affine}, the first line of~\eqref{int_bne} can be written as
\begin{equation} \label{bce1}
\E\qty[X(i)] = \int_0^1 w(i,j)\E\qty[X(j)]\d j + b(i)\mu_\theta + c(i), 
\quad \forall i \in [0,1].
\end{equation}
This expression decomposes the mean of each agent's action into the weighted average of others' mean actions, the response to the expected state, and the payoff-relevant constant term.
Note that this characterization of the mean is independent of the information structure and depends only on the payoff structure $(w,b,c)$ and $\mu_\theta$.

Similarly, the second line of~\eqref{int_bne} restricts the second moments of the induced outcome.
Multiplying both sides of that equation by $\varphi(i)^\top$ and applying the covariance formulae in Lemma~\ref{lem_affine} yields
\begin{equation} \label{bce2}
\Var[X(i)] = \int_0^1 w(i,j)\Cov[X(i),X(j)]\d j + b(i)\Cov[X(i),\theta], 
\quad \forall i \in [0,1].
\end{equation}
Hence, each agent's action volatility is decomposed into two parts: 
one arising from correlations with others' actions and another from correlations with the state.
Since any affine strategy profile induces a jointly Gaussian distribution over the state and actions, the moment restrictions~\eqref{bce1} and~\eqref{bce2} characterize the set of equilibrium outcomes that can be induced under some information structure.
Following \citet{bm2013}, we refer to such a distribution as a \emph{Bayes correlated equilibrium (BCE)}.

\begin{definition}
A Gaussian process $\calX = \{\theta\} \cup \{X(i)\}_{i \in [0,1]}$ forms a \emph{Gaussian BCE} under a payoff structure $(w,b,c)$ if its first and second moments satisfy the restrictions~\eqref{bce1}~and~\eqref{bce2}.
\end{definition}

The next result shows that any induced outcome arising from an affine equilibrium of an LQG game must form a Gaussian BCE.

\begin{lemma} \label{lem_bce}
For any LQG game $\calG = (w,b,c,\calP)$ and its affine equilibrium $\bar{\varphi}$, the induced outcome of $(\calG, \bar{\varphi})$ forms a Gaussian BCE under $(w,b,c)$.
\end{lemma}

In fact, an even stronger statement holds.
The restrictions~\eqref{bce1} and~\eqref{bce2} must be satisfied by any equilibrium, possibly non-affine, and under any information structure, even when the state and signals are not jointly Gaussian, as long as the payoff structure is quadratic.
This is because these restrictions are derived directly from the law of iterated expectations, which holds independently of distributional assumptions.
The proof of the lemma in Appendix~\ref{app_moment} clarifies this point.

Moreover, the converse of Lemma~\ref{lem_bce} holds true; namely, starting from any Gaussian BCE, one can construct a Gaussian information structure $\calP$ under which a unique affine equilibrium exists and reproduces the same first and second moments.
This statement is implied by Theorem~\ref{thm_equiv}, where we show that such replication is possible even when attention is restricted to the class of canonical information structures.

\begin{example}[Continued]
In what follows, suppose that $\theta \sim \calN(0,1)$.
Then the first line of \eqref{int_bne} implies $\varphi_0\equiv 0$, and the second line implies that $\varphi$ solves
\begin{equation} \label{ex_int}
\varphi(i) + (1-\tau)\int_0^1 P(i,j)\varphi(j)\d j = (1-\tau)P_\theta(i).
\end{equation}
By Theorem~\ref{thm_well}, this equation admits a unique solution for almost every $\tau$ and $\calP$.
In fact, an even stronger conclusion holds in this example, i.e., an affine equilibrium exists for every $\tau$ and $\calP$ and is the unique BNE, including among non-affine strategy profiles.\footnote{This is because the present example features strategic substitutability and payoff symmetry across all firms; see Proposition~\ref{prop_unique_const}.}

The normalization $\theta\sim\calN(0,1)$ allows us to represent PS as the integral of firms' action variances:
\begin{equation} \label{eq_PS}
{\rm PS} = \frac{1}{2} \int_0^1 \Var\qty[X(i)] \d i = \frac{1}{2}  \int_0^1 \|\varphi(i)\|^2 \d i,
\end{equation}
where $\|\cdot\|$ denotes the Euclidean norm, and $\Var\qty[X(i)] = \|\varphi(i)\|^2$ holds by Lemma~\ref{lem_affine}.
In addition, the authority's expected tax revenue (TR), defined as the $\tau$-fraction of the gross price $\theta-Y$ times total output $Y$, takes a similar form:
\begin{align}
{\rm TR}
&=\tau \E\qty[(\theta-Y) \cdot Y]
=\tau \int_0^1 \E\qty[\theta X(i)] \d i - \tau\int_0^1 \int_0^1 \E\qty[X(i)X(j)] \d i \d j \nonumber \\
&=\tau \int_0^1 \qty(\Cov\qty[X(i), \theta]-\int_0^1 \Cov\qty[X(i), X(j)]\d j)\d i \nonumber \\
&=\frac{\tau}{1-\tau} \int_0^1 \Var\qty[X(i)]\d i
=\frac{\tau}{1-\tau} \int_0^1 \|\varphi(i)\|^2 \d i, \label{eq_TR}
\end{align}
where the third line follows from \eqref{bce2} by recalling $b(i) = - w(i,j) = 1-\tau$ in the present setup.
\end{example}

\subsection{Relation to Finite-population Models}
\label{sec_nstep}

To conclude this section, we point out that the continuity result in Theorem~\ref{thm_well} bridges an affine equilibrium in the continuum-population model with equilibria in finite-population models by arguing that the former can be obtained as the limit of the latter as the number of agents tends to infinity.

To formalize this idea, let the set of agents $[0,1]$ be partitioned into $n$ equal-measure cells
\[
t_1 = \left[0,\, \frac{1}{n}\right] \quad \text{and} \quad t_k = \left(\frac{k-1}{n},\, \frac{k}{n}\right], \quad \forall k=2,\ldots,n.
\]
A payoff structure $(w,b,c)$ is \emph{$n$-step} if, for each $i,j,i',j' \in [0,1]$ with $i,j \in t_k$ and $i',j' \in t_{k'}$,
\[
w_{kk'} \equiv w(i,i') = w(j,j'), \quad b_k \equiv b(i) = b(j), \quad c_k \equiv c(i) = c(j).
\]
Likewise, an information structure is \emph{$n$-step} if, for each $i,j \in t_k$,
\[
S_k \equiv S(i) = S(j) \quad \text{a.s.}
\]
An LQG game is \emph{$n$-step} if both its payoff structure and information structure are $n$-step.

In an $n$-step LQG game, agents in the same cell entail the same payoff parameters and receive the same signal.
Hence, the best-response condition \eqref{eq_l} implies that $X_k \equiv X(i) = X(j)$ a.s.~whenever $i,j \in t_k$.
Moreover, for any strategy profile such that agents in the same cell adopt the same strategy, \eqref{eq_l} reduces to
\begin{equation} \label{br_finite}
X_k =\sum_{k'=1}^n  \frac{w_{kk'}}{n} \E\qty[X_{k'} \mid S_k] + b_k \E\qty[\theta \mid S_k] + c_k.
\end{equation}
This best-response system is identical to the one arising in a finite-population analogue of the LQG game with the finite set of agents $T=\{t_1,\ldots,t_n\}$, where each agent $t_k$ receives signal $S_k$ about the state $\theta$, and the complete-information best response takes the form $x_k = \sum_{k'=1}^n \frac{w_{kk'}}{n} x_{k'} + b_k \theta + c_k$.
Thus, the equilibrium of an $n$-step continuum-population game coincides with the equilibrium of its associated $n$-agent finite-population game, after identifying each cell $t_k$ with one finite-population agent.

Now, given a general continuum-population LQG game, consider a sequence of $n$-step approximations of its payoff and information structures.
Each $n$-step approximation admits an associated finite-population analogue, whose equilibrium can be embedded into the continuum model as a step-function strategy profile.
The continuity result in Theorem~\ref{thm_well} then implies that, as the $n$-step approximations converge to the original continuum-population game, the corresponding equilibria converge to the affine equilibrium of the original game.
In this sense, the continuum-population equilibrium studied in this paper is obtained as the limit of equilibria in finite-population analogues.


\section{Identification of Information Structures} \label{sec_id}

We now investigate the empirical content of equilibrium behavior in the LQG framework from the perspective of identification in econometrics \citep{manski1995}.
Suppose that our econometrician has access to the data on equilibrium outcomes, while she does not know precisely the primitives of an LQG game in play.
Her objective is to recover the underlying information structure from observables, at least, in a partial sense.


\subsection{Data and Hypothesis} \label{sec_asm}

The possibility of identification hinges on empirically relevant assumptions about the dataset and the structural model.
The econometrician faces a more difficult identification problem when the available data are limited and only weak restrictions are imposed on the structural model.
In what follows, we describe layers of assumptions considered in our analysis.

We begin by specifying what the econometrician observes.
Throughout, we assume that the econometrician knows the common prior over the state; that is, she knows that $\theta$ is ex ante normally distributed with mean $\mu_\theta$ and variance $\sigma^2_\theta$.
Since our interest lies in identifying agents' heterogeneous informational traits, we consider identification under rich individual choice data.
The strongest form of such data is the observability of the entire joint distribution of the state and actions.\footnote{Our joint observability is essentially equivalent to the assumption employed in the identification analysis of \cite{bm2013}. They assume that ``the econometrician observes the realized individual actions and the realized state. In other words, the econometrician learns the first and second moments of the joint equilibrium distribution over actions and state'' (see pp.\ 1290--1291). As such, since their model has a continuum of homogeneous agents who receive conditionally i.i.d.~signals, there is no essential difference between assuming that the econometrician observes the whole profile of realized actions and assuming that she learns the distribution of an individual agent's action.}
We also consider a weaker form of data, which requires only the observability of conditional action distributions at two distinct state realizations.

\begin{itemize}
\setlength{\leftmargin}{0pt}
\setlength{\itemindent}{0pt}
\item \emph{Joint observability}: The econometrician observes the joint distribution of $\calX$.
\item \emph{Conditional observability}: The econometrician observes the conditional action distributions $X(\cdot)|_{\theta = \hat{\theta}}$ for two distinct state realizations $\hat{\theta}_1,\hat{\theta}_2 \in \R$, which are also observable.
\end{itemize}

Since the common state prior is known to the econometrician, joint observability essentially postulates that the econometrician can observe the conditional action distribution $X(\cdot)|_{\theta=\hat{\theta}}$ for every state realization $\hat{\theta}\in\R$.
By contrast, under conditional observability, the econometrician has access to conditional action distributions only for two distinct state realizations.
We consider these assumptions of different strengths to sharpen our results: our impossibility result is established under joint observability, whereas our possibility results are established under conditional observability.

Next, we specify assumptions regarding the structural model---which we call \emph{hypotheses}---that describe the econometrician's prior knowledge about the underlying data-generating process.
Throughout, our most fundamental hypothesis is that the observed data are generated by Bayesian rational agents playing an affine equilibrium of some LQG game, while the econometrician does not know precisely which LQG game is being played.
Put differently, the econometrician faces a parametric identification problem, where the relevant parameters correspond to the primitives of an LQG game.
Our main interest lies in the identifiability of information structures $\calP$.

As with the observability assumptions, we consider several hypotheses of varying strength.
Depending on the result, the payoff structure is either known or unknown to the econometrician.
The information structure is always unknown to the econometrician, but the set of information structures deemed possible varies across the analysis.
Specifically, we consider either the set of all (Gaussian) information structures or the canonical class introduced in Definition~\ref{def_cano}.
Table~\ref{table_asm} summarizes the identification assumptions used for each result.

\begin{table}[t]
\caption{Identification assumptions and objects}
\renewcommand{\arraystretch}{1.1}
\begin{center}
{\footnotesize
\setlength{\tabcolsep}{4pt}
\begin{tabular}{|c|cccc|} \hline
 & Data & \begin{tabular}{c} Hypothesis on \\ payoff structure \end{tabular} & \begin{tabular}{c} Hypothesis on \\ information structure \end{tabular} & \begin{tabular}{c} Identification \\ object \end{tabular}  \\ \hline \hline
\begin{tabular}{c} Section~\ref{sec_equiv} \\ (Theorem~\ref{thm_equiv}) \end{tabular} & joint observability & known & general & $\calP$ \\ \hline
\begin{tabular}{c} Section~\ref{sec_cano} \\ (Theorem~\ref{thm_id}) \end{tabular} & conditional observability & unknown & canonical & $\calP^*$ \\ \hline
\begin{tabular}{c} Section~\ref{sec_var} \\ (Proposition~\ref{prop_var}) \end{tabular}  & conditional observability & unknown & general & \begin{tabular}{c} lower bound \\ on $r_i(\cdot \mid \calP)$ \end{tabular}  \\ \hline
\end{tabular}
}
\end{center}
\label{table_asm}
\end{table}

\begin{example}[Continued]
Hereafter, we discuss the example by viewing the tax authority as the econometrician.
As we have seen, PS and TR depend both on the chosen tax rate $\tau$ and on the underlying information structure, which is initially unknown.
Suppose that, before changing the tax rate to some new value $\tau$, the tax authority has access to rich data on firms' past production choices under the initial tax rate, say $\tau_0$.
Our observability assumptions require the distribution of each individual firm's choices to be observable, which presumes the availability of rich cross-sectional data.

The availability of such rich data may be compelling in connection with the $n$-step structure, where there are finitely many firm types and each type consists of infinitely many firms sharing the same payoff parameters.
Unlike the original setting in Section~\ref{sec_nstep}, however, suppose here that firms of the same type receive conditionally independent signals drawn from an identical distribution given the state.\footnote{One possible interpretation is that each type corresponds to a geographic or industrial group of firms, and firms in the same group may be exposed to the same source of policy announcements or news media, while the way such information is transmitted or perceived is subject to idiosyncratic noise.}
In the resulting equilibrium, firms of the same type choose different realized actions, but these actions are generated from the same conditional distribution.
Thus, observing many firms within the same type allows the econometrician to recover the action distribution associated with that type.
Conditional observability posits that the outcome of this game is observed for at least two distinct state realizations, or demand intercepts, which may correspond, for example, to different annual economic conditions.
Joint observability requires analogous observations for infinitely many state realizations.
In either case, the realized state is assumed to be observable.

Regarding the econometrician's hypotheses, it is natural in the present context to assume that the payoff structure is known to the authority.
More generally, however, our identification results speak to situations in which the authority does not know precisely how the demand shock $\theta$ and aggregate production $Y$ are translated into the market price.
\end{example}


\subsection{Canonical Information Structures}
\label{sec_equiv}

In this section, we impose no restriction on the class of information structures under consideration.
Instead, we assume that the econometrician observes the entire joint distribution of an equilibrium outcome and knows the precise payoff structure $(w,b,c)$.

The next definition introduces a class of simple information structures that will play a key role throughout the subsequent analysis.

\begin{definition} \label{def_cano}
An information structure $\calP^*$ is \emph{canonical} if each agent receives a one-dimensional signal $S^*(i)$ of the form
\[
S^*(i) = h(i) \cdot \theta + \varepsilon(i) \quad \text{with} \quad \Var\qty[S^*(i)] = 1,
\]
where $h:[0,1] \to \R_+$ is a given function and, together with a given function $g:[0,1]^2 \to \R$, the collection $\{\theta\} \cup \{\varepsilon(i)\}_{i \in [0,1]}$ forms a Gaussian process satisfying
\[
\E\qty[\varepsilon(i)] = 0, \quad
\Cov\qty[\varepsilon(i), \theta] = 0, \quad 
\Cov\qty[\varepsilon(i), \varepsilon(j)] = g(i,j).
\]
\end{definition}

\begin{remark}
In terms of the parametric specification of general information structures, a canonical information structure corresponds to the case with $\bar d=1$, where
\begin{equation} \label{P_cano}
P(i,j) = K(i,j) = \sigma_\theta^2 h(i)h(j) + g(i,j), \quad
P_\theta(i) = \frac{1}{\sigma_\theta}K_\theta(i) = \sigma_\theta h(i).
\end{equation}
Since the covariance matrix \eqref{cov_matrix_join} must be symmetric and positive semidefinite, $g$ must satisfy the corresponding symmetry and positive semidefiniteness requirements.
These requirements amount to conditions \eqref{canonical_psd2} and \eqref{canonical_psd3} in the next lemma.
In addition, the normalization $\Var\qty[S^*(i)] = 1$ is equivalent to condition \eqref{canonical_psd1}, and $h$ takes nonnegative values by construction.
These conditions are necessary and sufficient for the pair $(h,g)$ to represent a canonical information structure.
\end{remark}

\begin{lemma} \label{lem_canonical_psd}
A pair of functions $h:[0,1] \to \R_+$ and $g:[0,1]^2 \to \R$ represents a canonical information structure for which \eqref{cov_matrix_join} is symmetric and positive semidefinite if and only if the following conditions hold:
\begin{enumerate}[\rm(i)]
\item \label{canonical_psd1}
$\sigma_\theta^2h(i)^2+g(i,i)=1$ for all $i\in[0,1]$;
\item \label{canonical_psd2}
$g(i,j)=g(j,i)$ for all $i,j\in[0,1]$; and
\item \label{canonical_psd3}
for any finite set $\{i_1,\ldots,i_n\}\subseteq[0,1]$, the matrix 
\begin{equation} \label{g_psd}
\mqty[
g(i_1,i_1) & \cdots & g(i_1,i_n) \\
\vdots & \ddots & \vdots \\
g(i_n,i_1) & \cdots & g(i_n,i_n)]
\end{equation}
 is positive semidefinite.
\end{enumerate}
\end{lemma}

Under a canonical information structure, each agent receives a normalized one-dimensional signal $S^*(i)$, which is positively correlated with the state.
In particular, $S^*(i)$ admits an additive decomposition into two components: the \emph{common component}, $h(i)\cdot\theta$, and the \emph{idiosyncratic component}, $\varepsilon(i)$.\footnote{A similar decomposable process is considered in finance under the name of a $K$-factor model \citep[see, e.g.,][]{alnajjar1995}, but our canonical model is more permissive since agents' noise terms can be correlated.}
Since these two components are uncorrelated, the \emph{exposure function} $h:[0,1]\to\R_+$ represents the informativeness of agent $i$'s signal $S^*(i)$ about the state $\theta$.
On the other hand, we allow different agents' idiosyncratic components to be correlated with one another, and this correlation is captured by the \emph{idiosyncratic kernel} $g:[0,1]^2\to\R$.
Any canonical information structure is therefore summarized by two functions $h$ and $g$ satisfying Lemma~\ref{lem_canonical_psd}.
In light of this, we refer to the pair $(h,g)$ itself as a canonical information structure.

The next result shows that, for any Gaussian BCE under a given payoff structure, there exists a canonical information structure that, when paired with the same payoff structure, yields that Gaussian BCE as an induced outcome.
As a mnemonic, hereafter we use the superscript ``$*$'' to denote objects associated with a canonical information structure.

\begin{theorem} \label{thm_equiv}
Given a Gaussian BCE $\calX$ under a payoff structure $(w,b,c)$, there exist a canonical information structure $\calP^*$ and an affine equilibrium $\bar{\varphi}^*$ of the LQG game $\calG^* = (w,b,c,\calP^*)$ such that $\calX$ arises as the induced outcome of $(\calG^*,\bar{\varphi}^*)$.
\end{theorem}

Taken together, Lemma~\ref{lem_bce} and Theorem~\ref{thm_equiv} establish an equivalence between BNE and BCE.
That is, fixing a payoff structure $(w,b,c)$, any Gaussian process forms a Gaussian BCE if and only if it arises as the induced outcome of an affine equilibrium under some information structure.
We remark that this equivalence is not entirely new in itself but rather builds on related observations appearing in a sequence of works such as \cite{bm2013,bm2016,bhm2017}.

What we highlight here is its reinterpretation from an identification perspective.
Specifically, while Lemma~\ref{lem_bce} shows that the induced outcome under any information structure forms a Gaussian BCE, Theorem~\ref{thm_equiv} shows that, in reproducing such a Gaussian BCE, attention can be restricted to the class of canonical information structures.
This observation yields the following corollary.

\begin{corollary}
If a game $\calG = (w,b,c,\calP)$ has an affine equilibrium $\bar{\varphi}$, then there exist some canonical information structure $\calP^*$ and an affine equilibrium $\bar{\varphi}^*$ in $\calG^* = (w,b,c,\calP^*)$ such that the induced outcomes of $(\calG,\bar{\varphi})$ and $(\calG^*,\bar{\varphi}^*)$ coincide.
\end{corollary}

From a modeling perspective, this corollary implies that the modeler incurs no loss of generality in explaining any possible equilibrium outcome by assuming as if agents' signals are generated from some canonical information structure.
From an identification perspective, however, this points to the limit of fully identifying information structures.
That is, even when the strongest form of data, i.e., the entire joint distribution of $\calX$, is available to the econometrician, and even when she knows the fixed payoff structure $(w,b,c)$, any information structure is observationally equivalent to some canonical information structure in the sense that they induce the same outcome.

\begin{example}[Continued]
Consider the now-standard information structure $\calP$ as in \cite{ms2002}, where each agent receives two signals: a conditionally i.i.d.~private signal $S_1(i)= \theta+\alpha_1\eta_1(i)$ and a public signal $S_2\equiv S_2(i)=\theta+\alpha_2\eta_2$, where $\alpha_1,\alpha_2 > 0$ and $\eta_1(i)$, $\eta_2$, and $\theta$ are mutually independent standard normal random variables.
The unique affine equilibrium exhibits symmetry across firms and takes the form:\footnote{This way of writing the affine equilibrium differs slightly from \eqref{eq_a}, since $\phi_1$ and $\phi_2$ here capture responses to the raw signals $S_1(i)$ and $S_2$, rather than to standardized signals.}
\[
X(i) = \phi_1 S_1(i) + \phi_2 S_2 = \underbrace{\qty(\phi_1 + \phi_2)}_{\text{exposure to $\theta$}} \cdot \; \theta + \underbrace{\qty(\alpha_1\phi_1 \eta_1(i) + \alpha_2\phi_2 \eta_2)}_{\text{noise term}}.
\]
This outcome can be replicated by constructing a canonical information structure $\calP^* = (h,g)$ such that $h$ is proportional to the coefficient on $\theta$ in the above expression and $\varepsilon(i)$ accounts for the residual noise term.
The constructive proof of Theorem~\ref{thm_equiv} reveals that
\[
h(i) = \frac{\phi_1 + \phi_2}{D}, \quad
g(i,i) = \qty(\frac{\alpha_1\phi_1 + \alpha_2\phi_2}{D})^2, \quad 
g(i,j) = \qty(\frac{\alpha_2\phi_2}{D})^2,
\]
where $D = \sqrt{(\phi_1+\phi_2)^2 + \alpha_1^2\phi_1^2 + \alpha_2^2 \phi_2^2}$ represents the standard deviation of $X(i)$, and division by $D$ normalizes the variance of $S^*(i)$ to one.
It is worth noting that the construction of the observationally equivalent canonical information structure $\calP^*$ relies on the underlying payoff structure through the determination of $(\phi_1,\phi_2)$.
In particular, when $\tau$ changes, the induced outcome under the previously constructed $\calP^*$ need not coincide with that under $\calP$ anymore.
\end{example}

\subsection{Identification in the Canonical Class} \label{sec_cano}

In this section, we seek an affirmative resolution to our identification exercise by hypothesizing that the underlying information structure is canonical.
In other words, we treat the observed Gaussian BCE $\calX$ as if it were generated from the following statistical model, which follows from the affinity of strategies and the canonicality of signals:
\begin{align} \label{eq_id}
X(i) &= \varphi_0^* (i) + \varphi_1^*(i) \cdot \qty(S^*(i) - \E\qty[S^*(i)]) \nonumber \\
&= \varphi_0^*(i) + \varphi_1^*(i)h(i) \cdot \qty(\theta - \mu_\theta) + \varphi_1^*(i) \cdot \varepsilon(i), \quad \text{where} \quad \varepsilon (\cdot) \sim {\rm GP}(0,g).
\end{align}
Under this hypothesis, we consider the identifiability of the associated exposure function $h$ and idiosyncratic kernel $g$, as well as the intercept $\varphi_0^*$ and slope $\varphi_1^*$ of the corresponding affine equilibrium.

The following regularity conditions are considered to obtain sharp identification results.

\begin{assumption} \label{asm_var}
The observed Gaussian BCE $\calX$ satisfies $\Var\qty[X(i)] > 0$ for every $i \in [0,1]$.
\end{assumption}

\begin{assumption} \label{asm_beta}
Under two conditioning state realizations $\hat{\theta}_1, \hat{\theta}_2 \in \R$, the observed Gaussian BCE $\calX$ satisfies $\E[X(i) \mid \theta = \hat{\theta}_1] \neq \E[X(i) \mid \theta = \hat{\theta}_2]$ for every $i \in [0,1]$.
\end{assumption}

These assumptions postulate variation in agents' actions and serve as minimal regularity conditions for identification.
We note that Assumption~\ref{asm_beta} is stronger than Assumption~\ref{asm_var} since the law of total variance implies
\begin{equation} \label{law_var}
\Var\qty[X(i)] = \Var\qty[X(i) \mid \theta] + \Var\qty[\E\qty[X(i) \mid \theta]].
\end{equation}
Namely, the total variation in $X(i)$ decomposes into two parts: the variation in $X(i)$ that remains after conditioning on $\theta$, and the variation in the conditional mean of $X(i)$ attributed to variation in the state.
Assumption~\ref{asm_var} requires at least one of these components to be positive.
By contrast, Assumption~\ref{asm_beta} requires the second component to be positive, so that agent $i$'s action varies on average across different state realizations.

The key idea behind identification comes from the reduced-form representation \eqref{eq_id}.
Viewed as a simple linear regression model, the orthogonality between the state $\theta$ and the noise term $\varepsilon(i)$ implies that the coefficients $\varphi_0^*(i)$ and $\varphi_1^*(i)h(i)$ are identified by regressing $X(i)$ on $\theta$.
Specifically, conditional observability allows the econometrician to compute the conditional action mean $\hat{X}_t(i)=\E[X(i)\mid \theta=\hat{\theta}_t]$ for each $i$ and $t\in\{1,2\}$.
The coefficient $\varphi_1^*(i)h(i)$ is then recovered as the slope of the conditional action mean with respect to the state:
\begin{equation} \label{eq_beta}
\hat{\beta}(i)
\equiv \frac{\hat{X}_1(i) - \hat{X}_2(i)}{\hat{\theta}_1 - \hat{\theta}_2}
= \frac{\E[X(i) \mid \theta = \hat{\theta}_1] - \E[X(i) \mid \theta = \hat{\theta}_2]}{\hat{\theta}_1 - \hat{\theta}_2}.
\end{equation}
In words, $\hat{\beta}(i)$ represents how much agent $i$ changes her action when the state realization changes by one unit.
Crucially, the computation of $\hat{\beta}(i)$ requires only two conditioning state realizations, where we benefit from the linear-quadratic structure of the basic game.

To identify the correlation structure of idiosyncratic noise, the relevant statistic is the conditional action covariance given the state:
\begin{equation} \label{eq_gamma}
\hat{\gamma}(i,j) \equiv \Cov\qty[X(i), X(j) \mid \theta = \hat{\theta}_1]
= \E \qty[(X(i) - \hat{X}_1(i))(X(j) - \hat{X}_1(j)) \mid \theta = \hat{\theta}_1].
\end{equation}
This statistic $\hat{\gamma}(i,j)$ represents the covariance between agents $i$ and $j$'s actions conditional on the state realization, and hence captures variation coming solely from the idiosyncratic noise terms.
It is worth noting that $\hat{\gamma}(i,j)$ does not depend on the particular state realization, where we benefit from the property of Gaussian distributions that conditional covariances do not depend on the realized value of the conditioning variable.

Incidentally, $\hat{\beta}(i)$ and $\hat{\gamma}(i,j)$ correspond to the cross terms $\varphi_1^*(i)h(i)$ and $\varphi_1^*(i)\varphi_1^*(j)g(i,j)$ in \eqref{eq_id}.
Also, Assumption~\ref{asm_var} ensures that at least one of $\hat{\beta}(i)$ or $\hat{\gamma}(i,i)$ is nonzero, whereas Assumption~\ref{asm_beta} ensures that $\hat{\beta}(i) \neq 0$.
It is worth noting that these key statistics are recovered directly from the data based on the reduced-form model \eqref{eq_id}, rather than being determined through the structural model \eqref{int_bne}.
This explains why the econometrician can remain agnostic about the underlying payoff structure to recover these statistics.
The final step is to separate $h$ and $g$ from the identified cross terms by removing the equilibrium slope $\varphi_1^*$.
This is accomplished by using condition~\eqref{canonical_psd1} in Lemma~\ref{lem_canonical_psd}, which follows from the normalization of the canonical signal variance.

In the statement of the following theorem, we write $\sgn(r)$ for the sign of a real number $r$, i.e., $\sgn(r)=1$ if $r>0$, $\sgn(r)=-1$ if $r<0$, and $\sgn(r)=0$ if $r=0$.

\begin{theorem} \label{thm_id}
Given a Gaussian BCE $\calX$ under some unknown payoff structure, the following identification results hold under conditional observability.
\begin{enumerate}[\rm (i)]
\item \label{thm_id1}
The equilibrium intercept is identified as $\varphi_0^*(i) = \hat{X}_1(i) - \hat{\beta}(i)(\hat{\theta}_1 - \mu_\theta)$.

\item \label{thm_id2}
Under Assumption~\ref{asm_var}, $h(i)$ and $g(i,i)$ are identified as
\[
h(i) = \frac{|\hat{\beta}(i)|}{\sqrt{\sigma^2_\theta \hat{\beta}(i)^2 + \hat{\gamma}(i,i)}}, \quad
g(i,i) = \frac{\hat{\gamma}(i,i)}{\sigma^2_\theta \hat{\beta}(i)^2 + \hat{\gamma}(i,i)}.
\]
Moreover, the absolute values of $\varphi_1^*(i)$ and $g(i,j)$ for $i \neq j$ are identified, with formulae obtained by taking the absolute values of the corresponding expressions in Part~\eqref{thm_id3}.

\item \label{thm_id3}
Under Assumption~\ref{asm_beta}, $\varphi_1^*(i)$ and $g(i,j)$ for $i \neq j$ are identified as
\[
\varphi_1^*(i) = \sgn(\hat{\beta}(i))\sqrt{\sigma^2_\theta \hat{\beta}(i)^2 + \hat{\gamma}(i,i)}, \quad
g(i,j) = \frac{\sgn(\hat{\beta}(i)\hat{\beta}(j)) \cdot \hat{\gamma}(i,j)}{\sqrt{(\sigma^2_\theta \hat{\beta}(i)^2 + \hat{\gamma}(i,i))(\sigma^2_\theta \hat{\beta}(j)^2 + \hat{\gamma}(j,j))}}.
\]
\end{enumerate}
\end{theorem}

Together with the preceding result, this theorem implies that canonical information structures are parsimonious in two senses.
On the one hand, Theorem~\ref{thm_equiv} shows that this class can rationalize any Gaussian BCE that may arise under some information structure.
On the other hand, Theorem~\ref{thm_id} shows that, within this class, the information structure is uniquely determined under mild conditions on observables.
Specifically, the point identification of $(h,g,\varphi_0^*,\varphi_1^*)$ obtains under Assumption~\ref{asm_beta}, which requires each agent's action to be responsive to state realizations.

Though this assumption is arguably weak, Theorem~\ref{thm_id} still provides partial identification results even without it.
First, the equilibrium intercept $\varphi_0^*$ is fully identified without any assumption.
Indeed, noticing that $\varphi_0^*(i)$ represents the conditional mean action evaluated at the prior mean $\mu_\theta$, the identified formula has a simple interpretation:
$\varphi_0^*(i)$ is recovered by taking the conditional mean action $\hat{X}_1(i)$ at $\hat{\theta}_1$ and adjusting it back by the deviation $\hat{\theta}_1-\mu_\theta$, multiplied by the slope $\hat{\beta}(i)$. 

When Assumption~\ref{asm_var} is violated, \eqref{eq_id} implies that $\varphi_1^*(i)=0$ under any canonical information structure consistent with the observed Gaussian BCE $\calX$.
In this case, identification of other parameters is hopeless, since agent $i$'s action is constant across signal realizations, and the data contain no information that separates the state component of the signal from the noise component.
By contrast, when Assumption~\ref{asm_var} holds, these two components can be separated using the variation in actions across state realizations, as captured by $\hat{\beta}(i)$, and the residual conditional variation orthogonal to the state, as captured by $\hat{\gamma}(i,i)$.
In particular, since both $h(i)$ and $g(i,i)$ are nonnegative by definition, their values are exactly identified.
The identified formulae suggest that $h(i)$ and $g(i,i)$ are increasing in $|\hat{\beta}(i)|$ and $\hat{\gamma}(i,i)$, respecitvely.

Assumption~\ref{asm_var} still leaves open the possibility that $h(i)=0$, in which case agent $i$'s signal is totally uninformative about the state.
Then any variation in agent $i$'s action must stem from correlations between her signal and others' signals.
In this case, $|g(i,j)|$ can still be identified from the conditional action covariance orthogonal to the state, as captured by $\hat{\gamma}(i,j)$, but its sign remains ambiguous because there is no hypothesis that disciplines the sign of $g(i,j)$.
Under Assumption~\ref{asm_beta}, however, we have $h(i)>0$, which allows us to recover the equilibrium slope $\varphi_1^*(i)$ from $\hat{\beta}(i)$, including its sign.
Moreover, once $\varphi_1^*(i)$ is identified for all agents, we know precisely how agents react to their signals, and the idiosyncratic kernel $g(i,j)$ is also fully identified from the conditional action covariance $\hat{\gamma}(i,j)$.

To sum up, the canonical information structure is identified by using the key statistics $\hat{\beta}$ and $\hat{\gamma}$.
It is worth emphasizing that computing these statistics does not fully require observing the exact state realizations: $\hat{\beta}(i)$ is obtained by dividing the change in conditional action means, $\hat{X}_1(i)-\hat{X}_2(i)$, by the change in state realizations, $\hat{\theta}_1-\hat{\theta}_2$, and hence requires only the state difference rather than the exact values of $\hat{\theta}_1$ and $\hat{\theta}_2$ themselves.
Moreover, $\hat{\gamma}(i,j)$ requires no information about the state realization, since conditional covariances of Gaussian random variables do not depend on the realized value of the conditioning variable.
Therefore, although conditional observability postulates that state realizations are observable, the identification of $h$, $g$, and $\varphi_1^*$ requires only the difference $\hat{\theta}_1-\hat{\theta}_2$.\footnote{The exception is the equilibrium intercept $\varphi_0^*$, whose formula employs the deviation of $\hat{\theta}_1$ from $\mu_\theta$.
Since $\mu_\theta$ is known to the econometrician, this requires knowing the exact value of $\hat{\theta}_1$.
However, since $\varphi_0^*$ does not depend on the information structure and is instead determined by the payoff structure through the first line of \eqref{int_bne}, if the econometrician is assumed to know the payoff structure, then $\varphi_0^*$ can be recovered even without any data.}


\subsection{(Partial) Identification of Variance Reduction} \label{sec_var}

Our positive identification result relies on the as-if assumption that the information structure $\calP^*$ rationalizing the observed outcome belongs to the canonical class.
A remaining issue is that the ``true'' information structure $\calP$ may differ from $\calP^*$, even though the two are observationally equivalent under a given payoff structure.
In light of this, we now address how the identified canonical information structure $\calP^*$ is related to a general information structure $\calP$ that is observationally equivalent to it.

This question is important from a counterfactual viewpoint.
As long as we are concerned only with the equilibrium outcome under a fixed payoff structure, the distinction between $\calP$ and $\calP^*$ is immaterial, since they induce the same outcome by definition.
However, observational equivalence depends on the payoff structure.
In other words, once the payoff structure changes, two previously observationally equivalent information structures need not continue to induce the same outcome.
For example, when the authority changes the tax rate in the market competition example, it effectively changes the payoff structure of the game.
Then an information structure $\calP$ and its observationally equivalent canonical counterpart $\calP^*$ under the initial tax rate need not remain observationally equivalent after the policy change.
To anticipate policy effects on equilibrium outcomes, therefore, it is useful to understand how the identified canonical information structure is related to a general information structure that may be the true one.

To that end, our starting point is a rather trivial observation.
Each agent $i$'s signal $S(i)$ under any information structure contains at least as much information as the agent's action $X(i)$, simply because $X(i)$ must be measurable with respect to $S(i)$.
As such, the canonical signal $S^*(i)$ contains the minimal amount of information needed to implement $X(i)$, provided that the relevant non-degeneracy condition is satisfied.
Specifically, under Assumption~\ref{asm_var}, the equilibrium slope $\varphi_1^*(i)$ is nonzero, and thus \eqref{eq_id} implies a one-to-one relationship between $S^*(i)$ and $X(i)$.
Consequently, when a general information structure $\calP$ and a canonical information structure $\calP^*$ are observationally equivalent, the canonical signal $S^*(i)$ cannot contain more information than the general signal $S(i)$.

To formalize this idea, we measure how informative agent $i$'s signal is about a random variable $Z$, such as the state $\theta$, based on variance reduction.
Letting $S(i)$ be agent $i$'s signal from an information structure $\calP$, we define the statistic $r_i(Z \mid \calP)$ as
\begin{equation} \label{eq_def_var}
r_i(Z \mid \calP) \coloneqq \frac{\Var\qty[Z] - \E\qty[\Var\qty[Z \mid S(i)]]}{\Var\qty[Z]}.
\end{equation}
In words, $r_i(Z \mid \calP)$ measures the fraction of the variance of $Z$ that can be reduced by conditioning on agent $i$'s signal generated from $\calP$.
It is a natural measure of informativeness, especially in the present Gaussian environment, where the conditional variance does not depend on the signal realization.
Moreover, by viewing each agent's signal as a statistical experiment about $Z$, the result of \cite{hansen1974} implies that the informativeness ranking induced by variance reduction coincides with the ranking of \cite{blackwell1951}, which provides further rationale for using variance reduction to measure the informativeness of agents' signals.

\begin{lemma} \label{lem_var}
Suppose that two LQG games $(w,b,c,\calP)$ and $(w,b,c,\calP^*)$, together with their respective affine equilibria, induce the same outcome $\calX$, where $\calP$ is an arbitrary information structure and $\calP^*$ is a canonical information structure.
If $\Var\qty[X(i)]>0$, then $r_i(Z\mid\calP) \ge r_i(Z\mid\calP^*)$ holds for any random variable $Z$ with finite variance.
\end{lemma}

This lemma follows from two observations.
First, as discussed above, each agent's action must be measurable with respect to the agent's signal, while the converse is also true when the signal is canonical and $\Var\qty[X(i)]>0$.
Second, the law of total variance implies that conditioning on a more informative random variable weakly reduces conditional variance.
Combining these, agent $i$'s signal under $\calP$ must induce a weakly larger variance reduction of any random variable $Z$ than the canonical signal under $\calP^*$.\footnote{Lemma~\ref{lem_var} can be generalized in two directions, although these are not substantial for our subsequent analysis. First, the two information structures in Lemma~\ref{lem_var} can be coupled with possibly different payoff structures, as long as the two LQG games induce the same outcome. Second, a similar comparison of variance reduction holds not only for a single agent but also for any finite set of agents $N \subseteq [0,1]$. Specifically, defining $r_N(Z \mid \calP)$ by using $\{S(i)\}_{i \in N}$ as the conditioning event in \eqref{eq_def_var}, we have $r_N(Z \mid \calP) \ge r_N(Z \mid \calP^*)$, provided that $\Var\qty[X(i)]>0$ for all $i \in N$.}

Coupled with the identification result for canonical information structures, Lemma~\ref{lem_var} allows us to identify a lower bound on the state variance reduction for each agent under a general information structure.
Specifically, taking the state $\theta$ in the position of $Z$, we have
\begin{align*}
r_i\qty(\theta \mid \calP)
&= \frac{\sigma^2_\theta - \Var\qty[\theta \mid S(i)]}{\sigma^2_\theta}
= \frac{\sigma^2_\theta - \qty(\sigma^2_\theta - K_\theta(i)^\top K(i,i)^{-1} K_\theta(i))}{\sigma^2_\theta} \\
&= \qty(\frac{K(i,i)^{-\frac{1}{2}}K_\theta(i)}{\sigma_\theta})^\top \qty(\frac{K(i,i)^{-\frac{1}{2}}K_\theta(i)}{\sigma_\theta})
= P_\theta(i)^\top P_\theta(i)
= \|P_\theta(i)\|^2.
\end{align*}
In particular, the last expression equals $\sigma^2_\theta h(i)^2$ for a canonical information structure $\calP^*$.
Then substituting the identified formula for $h(i)$ from Theorem~\ref{thm_id} yields a lower bound on $r_i(\theta \mid \calP)$ that holds across all information structures $\calP$ that are observationally equivalent to $\calP^*$.
The next proposition reports this bound and the analogous bound for $r_i(X(j) \mid \calP)$.\footnote{The proof of Proposition~\ref{prop_var} shows that $r_i(X(j) \mid \calP) = \frac{\|P(i,j)\varphi(j)\|^2}{\Var\qty[X(j)]}$ holds in an affine equilibrium under $\calP$.}

\begin{proposition} \label{prop_var}
Let $\calX$ be an induced outcome of an arbitrary LQG game $\calG=(w,b,c,\calP)$ and its affine equilibrium $\bar{\varphi}$.
Given $\hat{\beta}$ and $\hat{\gamma}$ identified as in \eqref{eq_beta} and \eqref{eq_gamma} under conditional observability, if $\Var[X(i)] > 0$, then
\begin{gather}
r_i\qty(\theta \mid \calP) = \|P_\theta(i)\|^2 \ge \frac{\sigma^2_\theta \hat{\beta}(i)^2}{\sigma^2_\theta \hat{\beta}(i)^2 + \hat{\gamma}(i,i)}, \label{id_var1} \\
r_i\qty(X(j) \mid \calP) \propto \|P(i,j)\varphi(j)\|^2 \ge \frac{\qty(\sigma^2_\theta \hat{\beta}(i) \hat{\beta}(j) + \hat{\gamma}(i,j))^2}{\sigma^2_\theta \hat{\beta}(i)^2 + \hat{\gamma}(i,i)}. \label{id_var2}
\end{gather}
In particular, these inequalities hold with equality if $\calP$ is one-dimensional, i.e., $\bar{d}=1$.
\end{proposition}

The first inequality \eqref{id_var1} delivers a lower bound on state variance reduction that holds across all information structures $\calP$ that rationalize the observed Gaussian BCE.
A useful feature of this inequality is that its left-hand side is given by the primitive $\|P_\theta(i)\|^2$, which describes the strength of the correlation between agent $i$'s signal and the state.
Thus, the inequality directly speaks to a property of the information structure itself by identifying the minimum amount of information about the state that each agent must possess.
As discussed in Section~\ref{sec_actvar}, this property is useful for robust counterfactual prediction of aggregate action volatility following the change in payoff structures.

The second inequality \eqref{id_var2} delivers a lower bound on how much agent $i$ would have predicted agent $j$'s action in the observed equilibrium.
Although this result is less directly portable across counterfactual payoff structures because its left-hand side still contains agent $j$'s equilibrium slope, it can have some implications when coupled with the continuity result in Theorem~\ref{thm_well}.
Namely, since $\varphi$ changes continuously in response to changes in the payoff structure, if the payoff perturbation is sufficiently small, then a lower bound close to the right-hand side of \eqref{id_var2} continues to bound how much agent $i$ can predict agent $j$'s action under the changed payoff structure.

The equality statement shows that the one-dimensionality of the canonical information structure is what allows the canonical class to characterize the lower bound on variance reduction.
More generally, the next proposition provides a necessary and sufficient condition for there to be no gap between state variance reduction under an arbitrary information structure $\calP$ and the identified lower bound in \eqref{id_var1}.

\begin{proposition} \label{prop_var_eq}
Let $\calX$ be an induced outcome of an arbitrary LQG game $\calG=(w,b,c,\calP)$ and its affine equilibrium $\varphi$.
Given that $\Var[X(i)] > 0$, \eqref{id_var1} holds with equality if and only if $\varphi(i)$ and $P_\theta(i)$ are linearly dependent.
In particular, this condition is satisfied if at least one of $b(j)$, $w(i,j)$, $P_\theta(j)$, or $P(i,j)$ is zero for almost every $j \in [0,1] \setminus \{i\}$.
\end{proposition}

The equality condition is that agent $i$'s equilibrium slope is proportional to the signal-state correlation vector.
This condition is automatically satisfied when $\bar{d}=1$, provided that $\varphi(i)\neq 0$.
For multi-dimensional information structures, however, a gap may arise between the identified lower bound and the actual state variance reduction under $\calP$.

The source of this gap can be attributed to strategic interactions, which make agents care not only about the fundamental state but also about others' actions.
As a result, equilibrium strategies combine each agent's estimate of the state with her estimate of strategically relevant uncertainty, which complicates the econometrician's identification problem.
This can be seen directly from the second-moment restriction in \eqref{int_bne}, where the signal-state covariance vector $P_\theta(i)$ appears as the constant term on the right-hand side, while there is also the first integral term capturing strategic interaction with others.
If this integral term vanishes, then the linear dependence between $\varphi(i)$ and $P_\theta(i)$ is automatically satisfied.
This happens, for example, when $w(i,\cdot)$ is identically zero.
In this case, strategic interaction is absent, and equilibrium actions fully reveal each agent's best estimate of the fundamental state.
Several other cases in which this integral term vanishes are also collected in Proposition~\ref{prop_var_eq}.


\subsection{Robust Prediction of Action Volatility}
\label{sec_actvar}

We now showcase how the identified lower bound on variance reduction can be translated into a bound on equilibrium action volatility.
Given an LQG game $(w,b,c,\calP)$ and an affine equilibrium $X(\cdot)$ represented by $\bar{\varphi}=(\varphi_0,\varphi)$, recall from Lemma~\ref{lem_affine} that the action variance of agent $i$ is given by $\Var\qty[X(i)]=\|\varphi(i)\|^2$.
The slope $\varphi(i)$ is determined by the equilibrium condition \eqref{int_bne}, and hence depends on both the payoff structure $(w,b,c)$ and the information structure $\calP$.
The question we ask is how the aggregate action volatility $\int_0^1 \Var\qty[X(i)] \d i$ varies across payoff structures, given an information structure $\calP$ that is unknown but whose traits are partially identified by the previous results.

Our focus on aggregate action volatility is motivated by its direct welfare interpretation in a variety of LQG games, as highlighted in the literature on the social value of public information \citep{ms2002, ap2007,ui2015}.
Specifically, consider the payoff function
\[
u_i\qty(x(i), y(i), \theta) = \qty(y(i)+b(i)\theta+c(i)) \cdot x(i) - \frac{x(i)^2}{2},
\]
which corresponds to the case of \eqref{eq_q} where only the first row and first column of $U(i)$ can be nonzero, so that payoff externalities enter only through terms relevant to agent $i$'s action $x(i)$.
The best response for this payoff is precisely given by \eqref{eq_l}.
Then, using the law of iterated expectations, the equilibrium expected payoff of agent $i$ is calculated as
\[
\E\qty[u_i(X(i),Y(i),\theta)]
= \E\biggl[\underbrace{\E_i \qty[Y(i)+b(i)\theta+c(i)]}_{=X(i) \text{ by \eqref{eq_l}}} \; \cdot \; X(i) - \frac{X(i)^2}{2}\biggr]
= \frac{1}{2}\E\qty|X(i)|^2.
\]
This implies that the expected payoff is proportional to $\Var\qty[X(i)]$, up to a constant component independent of information.

To obtain a transparent bound on action volatility, we employ the simplifying assumption that the LQG game is \emph{payoff-symmetric}, i.e., $w(i,j)\equiv a$ for almost every $(i,j) \in [0,1]^2$, where $a<1$ as in \cite{ap2007} and \cite{bm2013}.
Although the payoff weight function is symmetric, we still allow for other payoff parameters, $b(i)$ and $c(i)$, and the information structure $\calP$ to be heterogeneous across agents.
The next proposition derives lower bounds on aggregate action volatility in terms of aggregate state variance resuction, whose lower bound can in turn be identified from Proposition~\ref{prop_var}.

\begin{proposition} \label{prop_actvar}
Let $\bar{\varphi}$ be an affine equilibrium in a payoff-symmetric LQG game $(w,b,c,\calP)$ with $w(i,j) \equiv a<1$ for almost every $(i,j) \in [0,1]^2$.
Denote by $q(i) = \sigma_\theta b(i)P_\theta(i) \in \R^{\bar{d}}$ for each $i \in [0,1]$.
\begin{enumerate}[\rm (i)]
\item If $0\le a<1$, then
\begin{equation} \label{eq_actvar1}
\int_0^1 \|\varphi(i)\|^2 \d i \ge \int_0^1 \|q(i)\|^2 \d i.
\end{equation}
The equality holds if $a=0$ or $\int_0^1 P(i,j)q(j) \d j = 0$ for almost every $i \in [0,1]$.
\item If $a<0$, then
\begin{equation} \label{eq_actvar2}
\int_0^1 \|\varphi(i)\|^2 \d i \ge \qty(\frac{1}{1-a})^2 \int_0^1 \|q(i)\|^2 \d i.
\end{equation}
The equality holds if $\int_0^1 P(i,j)q(j) \d j = q(i)$ for almost every $i \in [0,1]$.
\end{enumerate}
\end{proposition}

The proposition shows that aggregate action volatility cannot be arbitrarily small once agents have informative signals about the state.
When the game exhibits strategic complementarity, i.e., $0\le a<1$, agents' direct responsiveness to the state already implies the lower bound in \eqref{eq_actvar1}.
By contrast, when the game exhibits strategic substitutability, i.e., $a<0$, the lower bound is attenuated by the factor $(1-a)^{-2}$, as in \eqref{eq_actvar2}, reflecting the dampening force of substitutability \citep{bramoulle2014}.
These bounds are obtained by exploiting some intrinsic properties of $P$, which represents normalized covariance matrices of agents' signals.\footnote{These properties are certain forms of boundedness, symmetry, and positive semidefiniteness, as formally reported in Lemma~\ref{lem_P}; see also the equation~\eqref{eq_bound_P} below and the discussion following it.
\cite{miyashita_ui_lqgd} make use of these properties to address information design in an LQG environment.
Within the context of information design, our exercise in this section can be understood as characterizing the minimized value of aggregate action volatility obtained by manipulating $P$, which governs correlations across different agents' signals, subject to the constraint that each agent's signal contains at least as much information about the state as that under the identified canonical information structure.}

Combining Proposition~\ref{prop_actvar} with Proposition~\ref{prop_var} yields a lower bound on aggregate action volatility across different payoff structures.
In the context of our running example, this implies that the authority can learn about guaranteed levels of PS and TR following a change in the tax rate.

\begin{example}[Continued]
Suppose that, under the initial tax rate $\tau_0$, the authority has access to data satisfying conditional observability.
The authority can then identify the canonical information structure rationalizing the observed market outcome under $\tau_0$.
By Proposition~\ref{prop_var}, this also allows the authority to identify a lower bound on each firm's state variance reduction.
Aggregating across firms gives
\[
\int_0^1 \|P_\theta(i)\|^2 \d i \ge \hat{\rho} \equiv \int_0^1 \frac{\hat{\beta}(i)^2}{\hat{\beta}(i)^2 + \hat{\gamma}(i,i)} \d i.
\]
Here, $\hat{\rho}$ represents the identified lower bound on average state variance reduction across firms.
The equality holds if the true information structure is canonical.

At the same time, the authority's potential objectives, such as PS and TR, can be represented as functions of aggregate action volatility.
Specifically, \eqref{eq_PS} and \eqref{eq_TR} express PS and TR in terms of aggregate action volatility, while Proposition~\ref{prop_actvar} bounds it from below by state variance reduction.
Combining these bounds with the identified lower bound $\hat{\rho}$ yields the following robust guarantees that hold for any $\tau$:
\[
{\rm PS} \ge \qty(\frac{1-\tau}{2-\tau})^2 \cdot \hat{\rho}
\quad \text{and} \quad
{\rm TR} \ge \frac{\tau(1-\tau)}{(2-\tau)^2} \cdot \hat{\rho}.
\]
These robust guarantees are useful for counterfactual analysis after a change in the tax rate.
As $\tau$ changes, PS and TR vary through the payoff structure, while the identified value of $\hat{\rho}$ serve as a constant fraction indicating how sensitively PS and TR would respond to changes in the tax rate at least.
\end{example}

Lastly, we discuss the equality conditions in Proposition~\ref{prop_actvar}.
The bound is tight when $a=0$, so that there is no strategic interaction and the game reduces to a prediction problem in which each agent is concerned only with matching her action to the state.\footnote{See, e.g., the application in Section~8.2 of \cite{smolinyamashita2026} for an example of such a prediction game in an LQG environment.}
Together with the equality condition in Proposition~\ref{prop_var_eq}, this degenerate case allows the econometrician to perfectly infer action volatility based on state variance reduction, which can be identified without any gap.
Therefore, the absence of strategic interaction substantially simplifies the identification exercise.

When strategic interaction is present, the equality condition requires a particular alignment between the signal correlation structure $P(i,j)$ and the signal-state component $q(i)=\sigma_\theta b(i)P_\theta(i)$.
This point can be illustrated from the fact that $P:[0,1]^2 \to \R^{\bar{d}\times \bar{d}}$ behaves like a correlation matrix for a continuum of random vectors.
Specifically, Lemma~\ref{lem_P} in Appendix~\ref{app_pf} implies that
\begin{equation} \label{eq_bound_P}
0 \le \int_0^1 \int_0^1 \psi(i)^\top P(i,j) \psi(j) \d j \d i \le \int_0^1 \|\psi(i)\|^2 \d i,
\end{equation}
for all square-integrable functions $\psi:[0,1]\to\R^{\bar d}$.\footnote{The first inequality says that $P$ is positive semidefinite, while the second says that the numerical radius of $P$ is bounded by one. In addition, we can show that $P$ is self-adjoint and its eigenvalues are contained in $[0,1]$. These properties parallel symmetry, positive semidefiniteness, and boundedness of a correlation matrix; see, e.g., Chapter~11 of \cite{brezis} for further reference.}

Each equality condition in Proposition~\ref{prop_actvar} describes the least favorable correlation structure for generating equilibrium action volatility among those satisfying \eqref{eq_bound_P}.
Specifically, in the complementarity case, equality in \eqref{eq_actvar1} requires the lower bound in \eqref{eq_bound_P} to be attained with $q$.
Intuitively, this means that agents' state-relevant information are uncorrelated through $P(i,j)$, which eliminates the channel through which strategic complementarity amplifies action volatility.
By contrast, in the substitutability case, equality in \eqref{eq_actvar2} instead requires the upper bound in \eqref{eq_bound_P} to be attained with $q$.
Intuitively, this means that agents' state-relevant information are maximally correlated through $P(i,j)$, which allows them to infer others' actions precisely.
As a result, the dampening force of strategic substitutability is strengthened, which in turn suppresses action volatility.


\section{Conclusion} \label{sec_conc}

This paper studies the identification of information structures from the distribution of equilibrium actions in a Bayesian game.
We offer a canonical class of information structures and show that this class is rich enough to rationalize any equilibrium outcome, yet parsimonious to admit point identification.
Moreover, a canonical information structure characterizes the minimal amount of information about the state that each agent possesses, as measured by variance reduction, across all observationally equivalent information structures.
The lower bound is tight when there are no strategic interactions among agents, but in general, there arises a strict gap because agents' strategic motives confound private information about fundamental and strategic uncertainty.
Still, the identified lower bound is useful for making predictions about equilibrium action volatility, which in turn shapes welfare measures in LQG games.
This provides a way to make counterfactual predictions about the effects of policy changes on economic outcomes.


\newpage

\appendix

\begin{center}
\LARGE{\bf Appendix}
\end{center}


\renewcommand{\thesection}{A}
\renewcommand{\theequation}{A\arabic{equation}}
\renewcommand{\thelemma}{A\arabic{lemma}}
\renewcommand{\theproposition}{A\arabic{proposition}}
\renewcommand{\thedefinition}{A\arabic{definition}}

\renewcommand{\theHlemma}{A\arabic{lemma}}
\renewcommand{\theHproposition}{A\arabic{proposition}}
\renewcommand{\theHequation}{A\arabic{equation}}

\setcounter{section}{0}
\setcounter{equation}{0}
\setcounter{lemma}{0}
\setcounter{proposition}{0}
\setcounter{definition}{0}

\section{Integration of Strategy Profiles}
\label{app_pettis}

In this Appendix~\ref{app_pettis}, the integral of a strategy profile is formally defined through the integral notion \`a la \cite{pettis1938}.
We are motivated to adopt this notion to sidestep the measurability issue (Remark~\ref{remark_int}); under our regularity conditions (Definition~\ref{def_st}), strategy profiles are indeed Pettis integrable.
Moreover, the properties of the Pettis integral enables us to characterize BNE through moment restrictions (Lemma~\ref{lem_bce}).
We also hint at how these properties can be leveraged to extend equilibrium uniqueness beyond affine strategy profiles.

\subsection{Pettis Integral}
\label{app_moment}

Consider the underlying probability space $(\Omega,\calF,\P)$ on which the state and signals are defined.
Let $\scrX$ be the Hilbert space of square-integrable random variables $X:\Omega\to\R$, with the inner product between $X,X' \in \scrX$ being $\E\qty[XX']$.
A \emph{(stochastic) process} refers to a mapping $X:[0,1]\to\scrX$ that assigns to each agent $i \in [0,1]$ a random variable $X(i) \in \scrX$.

\begin{definition}
A process $X:[0,1] \to \scrX$ is \emph{weakly measurable} if for any $Z \in \scrX$, the mapping $i \mapsto \E\qty[ZX(i)]$ is measurable.
Moreover, the process is \emph{Pettis integrable} if it is weakly measurable and for any measurable set $T \subseteq [0,1]$, there exists some random variable $Y_T \in \scrX$ such that
\begin{equation} \label{def_pettis}
\E\qty[ZY_T] = \int_T\E\qty[ZX(i)] \d i, \quad \forall Z \in \scrX.
\end{equation}
In this case, $Y_T$ is called the \emph{Pettis integral} of $X(\cdot)$ over $T$, which is often written as $\int_T X(i) \d i$.
\end{definition}

At the heart of the definition of the Pettis integral is a Fubini-like exchange property.\footnote{\cite{sun2006} develops an alternative approach to integrating a continuum of random variables by extending the usual product space over $[0,1]\times\Omega$.
Specifically, he introduces a Fubini extension, a probability space that extends the product space while preserving the Fubini property, namely the interchangeability of integration over $[0,1]$ and expectation over $\Omega$.
His approach also takes the Fubini property as a minimal desirable property for working with a continuum of random variables, and in this regard, the Pettis integral shares a similar spirit from an applied point of view.}
Specifically, taking any constant random variable for $Z$ in \eqref{def_pettis} gives
\begin{equation} \label{pettis_E}
\E \qty[\int_0^1 X(i) \d i] = \int_0^1 \E \qty[X(i)] \d i,
\end{equation}
which justifies the interchange of expectation and integration.
Similarly, taking $\int_0^1 X(i) \d i$ for $Z$ and using \eqref{def_pettis} twice, we obtain
\begin{equation} \label{pettis_S}
\E \qty|\int_0^1 X(i) \d i|^2 = \int_0^1 \int_0^1 \E\qty[X(i)X(j)] \d j \d i.
\end{equation}
In particular, applying \eqref{pettis_S} to the centered process $\tilde{X}(i) = X(i)-\E\qty[X(i)]$, we obtain the following covariance formula:
\begin{equation} \label{pettis_V}
\Var \qty[\int_0^1 X(i) \d i] = \int_0^1 \int_0^1 \Cov\qty[X(i),X(j)] \d j \d i.
\end{equation}

The calculations of PS and TR (see \eqref{eq_PS} and \eqref{eq_TR} in the running example) can now be justified by the above properties.
In addition, these properties are used to prove Lemma~\ref{lem_bce}, as shown below.

\begin{proof}[Proof of Lemma~\ref{lem_bce}]
Let $X(\cdot)$ be any BNE, so that the best-response condition \eqref{eq_l} is satisfied for every $i \in [0,1]$.
Taking unconditional expectations on both sides of \eqref{eq_l}, the law of iterated expectations and \eqref{pettis_E} readily imply the first moment restriction \eqref{bce1}.
Moreover, subtracting \eqref{bce1} from \eqref{eq_l}, we have
\[
X(i)-\E\qty[X(i)] = \E_i\qty[\int_0^1 w(i,j) \qty(X(j) - \E\qty[X(j)]) \d j] + b(i) \E_i \qty[\theta - \mu_\theta].
\]
Since $X(i)$ is measurable with respect to $S(i)$, multiplying both sides of this equation by $X(i)-\E[X(i)]$ yields
\begin{align*}
&\qty|X(i)-\E[X(i)]|^2 \\
&= \E_i\qty[\qty(X(i)-\E\qty[X(i)])\int_0^1 w(i,j)\qty(X(j) - \E\qty[X(j)]) \d j]+ b(i) \E_i \qty[\qty(X(i)-\E\qty[X(i)]) \qty(\theta - \mu_\theta)].
\end{align*}
Taking unconditional expectations on both sides, the law of iterated expectations and the defining property \eqref{def_pettis} of the Pettis integral, applied with $Z = X(i) - \E[X(i)]$, imply the second moment restriction \eqref{bce2}.
\end{proof}

Note that the above proof of Lemma~\ref{lem_bce} does not rely on $\calX$ being a Gaussian process.
Hence, these moment restrictions are necessarily satisfied under any BNE and any information structure, as long as each agent $i$'s best-response condition is linear in her best estimates of $\theta$ and $Y(i)$, as in \eqref{eq_l}.

The next lemma shows that any strategy profile that satisfies the regularity conditions in Definition~\ref{def_st} is weakly measurable and Pettis integrable.\footnote{The proof of Lemma~\ref{lem_pettis} parallels mostly the arguments in Appendix~A.1 of \cite{alnajjar1995}, which establishes Pettis integrability for a continuum of i.i.d.~random variables, but a self-contained proof is provided here for the sake of completeness.}

\begin{lemma} \label{lem_pettis}
A process $X: [0,1] \to \scrX$ is weakly measurable if and only if the mapping $i \mapsto \E \qty[X(i)X(j)]$ is measurable for every $j \in [0,1]$.
If, in addition, the mapping $i \mapsto \E\qty|X(i)|^2$ is integrable, then $X(\cdot)$ is Pettis integrable.
\end{lemma}

\begin{proof}

It is evident that if $X$ is weakly measurable, then the mapping $i \mapsto \E \qty[X(i)X(j)]$ is measurable for every $j \in [0,1]$.

Conversely, suppose that $i \mapsto \E \qty[X(i)X(j)]$ is measurable for every $j \in [0,1]$.
Let $\scrW \subseteq \scrX$ be the closure of the finite linear span of $\{X(i)\}_{i \in [0,1]}$.
Also, let $\scrW^\bot$ be the orthogonal complement of $\scrW$, i.e., the set of all $Y \in \scrX$ such that $\E\qty[YW] = 0$ for all $W \in \scrW$.
By Theorem~3 on p.~55 of \cite{lax}, for any $Z \in \scrX$, there exists a unique orthogonal decomposition $(W,Y) \in \scrW \times \scrW^\bot$ such that $Z = W+Y$.
Since $\E\qty[ZX(i)] = \E\qty[WX(i)]$ for all $i \in [0,1]$, it is without loss of generality to assume that $Z \in \scrW$ to verify weak measurability.
By the construction of $\scrW$, there exists a sequence $\{W_n\}_{n \in \N}$ with $\lim_{n \to \infty}\E\qty|Z-W_n|^2 = 0$ such that each $W_n$ is represented as a finite linear combination
\[
W_n = \sum_{k=1}^{K_n} \alpha_{n,k} X(j_{n,k}),
\]
for some $K_n \in \N$, $\alpha_{n,1},\ldots,\alpha_{n,K_n} \in \R$, and $j_{n,1},\ldots,j_{n,K_n} \in [0,1]$.
For each $n \in \N$, define $f_n:[0,1] \to \R$ by
\[
f_n(i) = \sum_{k=1}^{K_n} \alpha_{n,k} \E\qty[X(j_{n,k}) X(i)], \quad \forall i \in [0,1].
\]
By assumption, each $f_n$ is a measurable function.
Moreover, by construction, we have $\E\qty[W_nX(i)] = f_n(i)$.
Since convergence in norm implies weak convergence, it follows that
\[
\E \qty[ZX(i)] = \lim_{n \to \infty} \E\qty[W_nX(i)] = \lim_{n \to \infty} f_n(i), \quad \forall i \in [0,1].
\]
This shows that $i \mapsto \E\qty[ZX(i)]$ is the pointwise limit of measurable functions, and is therefore measurable by Lemma~4.29 of \cite{ab2006}.
Hence, $X(\cdot)$ is weakly measurable.

Now suppose, in addition, that the mapping $i \mapsto \E\qty|X(i)|^2$ is integrable.
For any $Z \in \scrX$, the Cauchy--Schwarz inequality implies that
\[
\int_0^1 \qty|\E\qty[Z X(i)]| \d i 
\le \qty(\E\qty|Z|^2)^{\frac{1}{2}} \int_0^1 \qty(\E\qty|X(i)|^2)^{\frac{1}{2}} \d i
\le \qty(\E\qty|Z|^2)^{\frac{1}{2}} \qty(\int_0^1 \E\qty|X(i)|^2 \d i)^{\frac{1}{2}} < \infty,
\]
which shows that $i \mapsto \E\qty[Z X(i)]$ is integrable.
Hence, Theorem~11.52 of \cite{ab2006} implies that $X(\cdot)$ is Pettis integrable.
\end{proof}


\subsection{Uniqueness of BNE}
\label{app_unique}

Now, we hint at how the properties of Pettis integral can be used to extend equilibrium uniqueness beyond affine strategy profiles.
To this end, suppose that $X(\cdot)$ and $X'(\cdot)$ are two BNEs, which need not be affine.
Since both satisfy the best-response condition \eqref{eq_l}, taking the difference gives
\[
X(i)-X'(i) = \E_i \qty[Y(i) - Y'(i)] = \E_i \qty[\int_0^1 w(i,j) \qty(X(j) - X'(j)) \d j].
\]
Let $\xi(i) = X(i)-X'(i)$.
Because $\xi(i)$ must be measurable with respect to agent $i$’s information, multiplying both sides of the above equation by $\xi(i)$ yields
\begin{equation} \label{unique_pettis1}
\qty|\xi(i)|^2 = \E_i \qty[\xi(i)\int_0^1 w(i,j) \xi(j) \d j].
\end{equation}
Moreover, applying the law of iterated expectations yields
\[
\E\qty|\xi(i)|^2 = \E \qty[\xi(i) \int_0^1 w(i,j) \xi(j) \d j] = \int_0^1 w(i,j) \E\qty[\xi(i) \xi(j)] \d j,
\]
where the second equality uses the fact that $\int_0^1 w(i,j) \xi(j) \d j$ is defined as a Pettis integral and therefore its inner product with $\xi(i)$ enjoys the exchange property.
Then, integrating both sides over $i \in [0,1]$, we obtain
\begin{equation} \label{unique_pettis2}
\int_0^1 \E\qty|\xi(i)|^2 \d i = \int_0^1 \int_0^1 w(i,j) \E\qty[\xi(i)\xi(j)] \d j \d i.
\end{equation}

By construction, if $\E|\xi(i)|^2 = 0$ for all $i \in [0,1]$, then $\E\qty|\xi(i)|^2 = \E\qty|X(i)-X'(i)|^2 = 0$ for all $i \in [0,1]$, and thus the two BNEs $X(\cdot)$ and $X'(\cdot)$ coincide in the mean-square sense.
Since these BNEs are not assumed to be affine in this argument, this would imply that the affine equilibrium in Theorem~\ref{thm_well} is in fact the unique equilibrium, even among non-affine strategy profiles.
As such, the equation~\eqref{unique_pettis2} indicates when this stronger conclusion holds.
Specifically, since the left-hand side of \eqref{unique_pettis2} must be nonnegative, once the right-hand side is shown to be nonpositive, we obtain $\int_0^1 \E\qty|\xi(i)|^2 \d i = 0$.
This implies that $\E|\xi(i)|^2 = 0$ for all $i \in [0,1]$, and therefore the affine equilibrium is unique among all BNEs.\footnote{Although $\int_0^1 \E\qty|\xi(i)|^2 \d i = 0$ implies $\E\qty|\xi(i)|^2 = 0$ only for almost every $i$, since \eqref{unique_pettis1}---which follows from the best-response condition---must hold for every agent $i$, this in fact forces $\E|\xi(i)|^2 = 0$ for every $i$.}
The next results provides sufficient conditions on primitives under which this strong uniqueness holds.

\begin{proposition} \label{prop_unique_const}
If $w(i,j) \equiv a \le 0$ for almost every $(i,j) \in [0,1]^2$, then there exists an affine equilibrium, and moreover, it is the unique BNE, including non-affine strategy profiles.
\end{proposition}

\begin{proof}
Suppose that $w(i,j) \equiv a \le 0$ for almost every $(i,j) \in [0,1]^2$.
Then the left-hand side of \eqref{unique_pettis2} is nonnegative, whereas  \eqref{pettis_S} implies that the right-hand side is nonpositive.
Hence, the preceding discussion implies that the game has at most one BNE, including non-affine strategy profiles.
To establish existence, we invoke some results from the proof of Theorem~\ref{thm_well}.
When $w$ is constant and nonpositive almost everywhere, one can show directly, using Lemma~\ref{lem_P}~\eqref{lem_P3}, that all eigenvalues of the operator $T$ defined in \eqref{def_T} are nonpositive.
In particular, $1$ cannot be an eigenvalue of $T$, and hence Lemma~\ref{lem_alt} implies that a unique affine equilibrium exists.
\end{proof}

\begin{proposition} \label{prop_unique}
Given any $n$-step LQG game, let $W$ be an $n$-by-$n$ matrix whose $(k,k')$-entry is $w_{kk'}$.
If the matrix $\frac{W+W^\top}{2}$ is negative semidefinite, then the game has at most one BNE, including non-affine strategy profiles.
\end{proposition}

\begin{proof}
From \eqref{br_finite}, agents in the same cell $t_k$ adopt the same strategy in any BNE of an $n$-step LQG game.
This implies that $\xi_k \equiv \xi(i) = \xi(j)$ for any $i,j \in t_k$ and $1 \le k \le n$.
Let $V$ be an $n$-by-$n$ matrix whose $(k,k')$-entry is $\E\qty[\xi_k\xi_{k'}]$, and observe that $V$ is positive semidefinite.
Then, the right-hand side of \eqref{unique_pettis2} can be written as
\[
\int_0^1 \int_0^1 w(i,j) \E\qty[\xi(i)\xi(j)] \d j \d i
= \sum_{k=1}^n \int_{t_k} \qty(\sum_{k'=1}^n \int_{t_{k'}} w_{kk'}v_{kk'} \d j) \d i
= \frac{1}{n^2} \sum_{k=1}^n \sum_{k'=1}^n w_{kk'}v_{kk'}.
\]
Since $V$ is symmetric, the final expression is equivalently obtained by replacing $W$ with its symmetric part $\frac{W+W^\top}{2}$, which is  negative semidefinite by assumption.
Because $V$ is positive semidefinite, the Schur product theorem \citep[Theorem 7.5.3 of][]{horn} implies that the final expression is nonpositive.
Then, \eqref{unique_pettis2} yields $\int_0^1 \E|\xi(i)|^2 \d i = 0$, and the preceding argument implies uniqueness of BNE.
\end{proof}

Economically, the negative semidefiniteness of $W$ captures a form of strategic substitutability in the $n$-step game.
For instance, if $a \equiv w(i,j)$ is constant, as in the models of \cite{ap2007} and \cite{bm2013}, then the condition is satisfied when $a \le 0$.
More generally, one sufficient condition for $W$ to be negative semidefinite is that its diagonal entries are negative and satisfy a diagonal-dominance condition:
\[
w_{kk} \le 0 \quad \text{and} \quad \sum_{k' \neq k} |w_{kk'}| \le |w_{kk}|, \quad \forall k=1,\ldots,n.
\]
That is, each agent faces strategic substitutability against agents in the same cell, while strategic effects across cells are unrestricted in sign.
Thus, strategic substitutability and complementarity may coexist across cells, as long as those cross-cell effects remain sufficiently small in absolute value.

The key step in the proof of Proposition~\ref{prop_unique} is the representation of the right-hand side of \eqref{unique_pettis2} as the Hadamard product between a positive semidefinite matrix $V$ and a negative semidefinite matrix $W$, which, by the Schur product theorem, becomes nonpositive.
Beyond $n$-step LQG games, this intuition carries over by viewing $w(i,j)$ and $\E\qty[\xi(i)\xi(j)]$ as infinite-dimensional analogues of matrices.
This idea is formalized in follow-up work by \cite{miyashita_ui_large} based on an infinite-dimensional analogue of the Schur product theorem, where uniqueness of BNE is established under the condition that the integral kernel $w$ exhibits a certain form of negative semidefiniteness.


\renewcommand{\thesection}{B}
\renewcommand{\theequation}{B\arabic{equation}}
\renewcommand{\thelemma}{B\arabic{lemma}}
\renewcommand{\theproposition}{B\arabic{proposition}}

\renewcommand{\theHlemma}{B\arabic{lemma}}
\renewcommand{\theHproposition}{B\arabic{proposition}}
\renewcommand{\theHequation}{B\arabic{equation}}

\setcounter{section}{1}
\setcounter{equation}{0}
\setcounter{lemma}{0}
\setcounter{proposition}{0}

\section{Proofs of Main Results} \label{app_pf}
\setcounter{equation}{0}
\setcounter{definition}{0}
\setcounter{assumption}{0}
\setcounter{claim}{0}
\setcounter{lemma}{0}
\setcounter{proposition}{0}

In Appendix~\ref{app_pf}, we provide all proofs omitted from the main text.
We begin by summarizing the notation used throughout the appendix, especially in the proof of Theorem~\ref{thm_well}.
All mathematical definitions are standard and can be found in \cite{lax}, \cite{brezis}, or \cite{krees}, which are the main references for the functional analysis tools used below.

Denote by $\R^{n}$ the space of $n$-dimensional column vectors and by $\R^{n\times n}$ the space of $n\times n$ matrices.
We use the Euclidean norm on $\R^{n}$, denoted by $\|v\| \coloneqq \sqrt{v^\top v}$.
The induced operator norm on $\R^{n \times n}$ is denoted by $\vertd{M} \coloneqq \sup_{v \in \R^{n} \setminus \{0\}} \frac{\|Mv\|}{\|v\|}$.
For $n,m \in \N$ and a set $A \subseteq \R^m$, we denote by $\calL_2^n(A)$ the Hilbert space of vector-valued functions $f,g:A \to \R^n$, equipped with the inner product
\[
\langle f,g \rangle \coloneqq \int_A f(x)^\top g(x) \d x.
\]
The induced norm is denoted by $\|f\| \coloneqq \sqrt{\langle f,f \rangle} = (\int_A \|f(x)\|^2 \d x)^{\frac{1}{2}}$.
Extending this definition to a matrix-valued function $F:A \to \R^{n\times n}$, we define its norm by $\|F\| \coloneqq (\int_A \vertd{F(x)}^2 \d x)^{\frac{1}{2}}$.
For a linear operator $T: \calL_2^n(A) \to \calL_2^n(A)$, we denote the induced operator norm by $\vertd{T} \coloneqq \sup_{\|f\|>0} \frac{\|Tf\|}{\|f\|}$.
Let $\Lambda(T)$ denote the set of real eigenvalues of $T$.
The identity operator is denoted by $I$.

In accordance with the above notation, the norms of the payoff parameters are given by
\[
\|b\| = \qty(\int_0^1 |b(i)|^2 \d i)^{\frac{1}{2}}, \quad
\|c\| = \qty(\int_0^1 |c(i)|^2 \d i)^{\frac{1}{2}}, \quad
\|w\| = \qty(\int_0^1 \int_0^1 |w(i,j)|^2 \d j \d i)^{\frac{1}{2}},
\]
all of which are assumed to be finite.
Also, the norms of the informational parameters are
\[
\|P\| = \qty(\int_0^1 \int_0^1 \vertd{P(i,j)}^2 \d j \d i)^{\frac{1}{2}}, \qquad
\|P_\theta\| = \qty(\int_0^1 \|P_\theta(i)\|^2 \d i)^{\frac{1}{2}},
\]
which are finite by construction, as shown in Lemma~\ref{lem_P} below.
Putting these objects together, we measure the distance between two LQG games with the same signal dimension $\bar{d}$ as follows:
\begin{equation} \label{norm_lqg}
\left\|\calG-\calG'\right\| \coloneqq \left\|b-b'\right\| + \left\|c-c'\right\| + \left\|w-w'\right\| + \left\|P-P'\right\| + \left\|P_\theta-P'_\theta\right\|.
\end{equation}

In Theorem~\ref{thm_well}, convergence of LQG games is understood with respect to the norm \eqref{norm_lqg}, and genericity refers to the corresponding notions of norm-openness and norm-denseness.
With this formalization in place, we can now formally state Theorem~\ref{thm_well} as follows.

\setcounter{theorem}{0}
\begin{theorem}[Restatement] \label{thm_well_re}
The following statements hold:
\begin{itemize}
\item {\rm (Denseness).} For any LQG game $\calG = (w,b,c,\calP)$ that does not have a unique affine equilibrium and every $\epsilon > 0$, there exists a payoff weight function $w'$ with $\|w-w'\| < \epsilon$ such that the LQG game $\calG' = (w',b,c,\calP)$ has a unique affine equilibrium.
\item {\rm (Openness and Convergence).} For any LQG game $\calG$ that has a unique affine equilibrium $\bar{\varphi}$, if $\{\calG_n\}_{n \in \N}$ is a sequence of LQG games such that $\lim_{n \to \infty} \|\calG - \calG_n\| = 0$, then $\calG_n$ has a unique affine equilibrium $\bar{\varphi}_n$ for every sufficiently large $n \in \N$. Moreover, we have $\lim_{n \to \infty} \|\bar{\varphi} - \bar{\varphi}_n\| = 0$.
\end{itemize}
\end{theorem}

Let us record some properties of information structures, viewed as a linear operator, that follow from the definitions of $P(i,j)$ and $P_\theta(i)$.
Since $P(i,j)$ and $P_\theta(i)$ are defined as standardized covariance matrices between different agents' signals and between the state and signals, respectively, they satisfy properties analogous to those of correlation coefficients.
In particular, their sizes are bounded by $1$, and they satisfy suitable symmetry and positive semidefiniteness properties.

\begin{lemma} \label{lem_P}
Consider any information structure $\calP = (P,P_\theta)$.
\begin{enumerate}[\rm (i)]
\item \label{lem_P1}
For any $i \in [0,1]$, we have $\|P_\theta(i)\| \le 1$.
\item \label{lem_P2}
For any $i,j \in [0,1]$, we have $\vertd{P(i,j)} \le 1$.
\item \label{lem_P3}
When $P$ is viewed as a linear operator $\varphi \mapsto \int_0^1 P(\cdot, j) \varphi(j) \d j$ over $\calL_2^{\bar{d}}[0,1]$, it is self-adjoint and positive semidefinite in the sense that $\langle \psi, P\varphi \rangle = \langle P\psi, \varphi \rangle$ and $\langle \varphi, P \varphi \rangle \ge 0$ hold for any $\psi,\varphi \in \calL_2^{\bar{d}}[0,1]$.
In particular, it holds that
\begin{equation} \label{eq_normP}
\|\varphi\|^2 \ge \langle \varphi, P\varphi \rangle \ge \|P\varphi\|^2 \ge 0.
\end{equation}
\end{enumerate}
\end{lemma}

\begin{proof}
First, since $K(i,i)$ and $\qty[\begin{smallmatrix}K(i,i) & K_\theta (i) \\ K_\theta(i)^\top & \sigma^2_\theta\end{smallmatrix}]$ are positive semidefinite by definition, the Schur complement $\sigma_\theta^2 - K_\theta(i)^\top K(i,i)^{-1} K_\theta(i)$ must be nonnegative.
Together with the definition of $P_\theta(i)$, this implies
\[
\|P_\theta(i)\|^2 \le \frac{K_\theta(i)^\top K(i,i)^{-1} K_\theta(i)}{\sigma^2_\theta} \le 1.
\]

Second, observe that $\qty[\begin{smallmatrix}I & P(i,j) \\ P(j,i) & I\end{smallmatrix}]$ arises as the covariance matrix of the standardized signals $K(i,i)^{-\frac{1}{2}}S(i)$ and $K(j,j)^{-\frac{1}{2}}S(j)$, and thus it must be positive semidefinite.
Taking the Schur complement, we see that $I-P(i,j)P(j,i)$ is positive semidefinite.
Moreover, since $P(j,i)=P(i,j)^\top$, this implies that all singular values of $P(i,j)$ are weakly less than $1$.
Hence, since the induced Euclidean matrix norm coincides with the largest singular value of a matrix \citep[see, e.g., p.\ 346 of][]{horn}, we have $\vertd{P(i,j)} \le 1$.

Third, we prove the claims about $P$ viewed as a linear operator on $\calL_2^{\bar{d}}[0,1]$.
Since $P(i,j)=P(j,i)^\top$, it holds for any $\varphi,\psi \in \calL_2^d[0,1]$ that
\begin{align*}
\langle \psi, P\varphi \rangle 
= \int_0^1 \psi(i)^\top \qty(\int_0^1 P(i,j)\varphi(j) \d j) \d i 
= \int_0^1 \qty(\int_0^1 P(j,i) \psi(i))^\top \varphi(j) \d j
= \langle P\psi, \varphi \rangle,
\end{align*}
from which $P$ is self-adjoint.
That $P$ is positive semidefinite follows directly from \eqref{eq_normP}.
To prove these inequalities, given any $\varphi \in \calL_2^d[0,1]$, we consider random variables $X(i) = \varphi(i)^\top K(i,i)^{-\frac{1}{2}}S(i)$.
Since $\Var\qty[X(i)] = \|\varphi(i)\|^2$ holds by Lemma~\ref{lem_affine}, we have
\[
\|\varphi\|^2 = \int_0^1 \|\varphi(i)\|^2 \d i = \int_0^1 \Var\qty[X(i)] \d i.
\]
Moreover, since $\Cov\qty[X(i), X(j)] = \varphi(i)^\top P(i,j) \varphi(j)$ again by Lemma~\ref{lem_affine}, we have
\[
\langle \varphi, P\varphi \rangle = \int_0^1 \int_0^1 \varphi(i)^\top  P(i,j) \varphi(j) \d j d i = \int_0^1 \int_0^1 \Cov\qty[X(i), X(j)] \d j \d i = \Var\qty[\int_0^1 X(i) \d i].
\]
where we use \eqref{pettis_V} for the third equality.
In particular, since $\int_0^1 \Var\qty[X(i)] \d i \ge \Var[\int_0^1 X(i) \d i]$ holds by the Cauchy--Schwarz inequality, it follows that $\|\varphi\|^2 \ge \langle \varphi, P\varphi \rangle$.
Lastly, to prove $\langle \varphi, P\varphi \rangle \ge \|P\varphi\|^2$, we shall show that
\begin{equation} \label{pf_P1}
\left\| \int_0^1 P(i,j) \varphi(j) \d j \right\|^2 \le \langle \varphi, P\varphi \rangle, \quad \forall i \in [0,1].
\end{equation}
Clearly, this inequality is satisfied if $\|\int_0^1 P(i,j) \varphi(j) \d j\| = 0$.
Otherwise, define a unit vector $\psi(i) = \frac{\int_0^1 P(i,j) \varphi(j) \d j}{\|\int_0^1 P(i,j) \varphi(j) \d j\|} \in \R^{\bar{d}}$ and a random variable $Z(i) = \psi(i)^\top K(i,i)^{-\frac{1}{2}} S(i)$.
By using the definition of the Pettis integral, we have
\begin{align}
&\Cov\qty[Z(i), \int_0^1 X(j) \d j] = \int_0^1 \Cov\qty[Z(i), X(j)] \d j \nonumber \\
&= \int_0^1 \psi(i)^\top K(i,i)^{-\frac{1}{2}} \Cov\qty[S(i), S(j)] K(j,j)^{-\frac{1}{2}} \varphi(j) \d j
= \int_0^1 \psi(i)^\top P(i,j) \varphi(j) \d j \nonumber \\
&= \psi(i)^\top \qty(\int_0^1 P(i,j) \varphi(j) \d j)
= \left\| \int_0^1 P(i,j) \varphi(j) \d j \right\|, \label{pf_P2}
\end{align}
where the last equality holds by the construction of $\psi(i)$.
On the other hand, noticing that $\Var[Z(i)] = \|\psi(i)\|^2 = 1$, the Cauchy--Schwartz inequality yields
\begin{equation} \label{pf_P3}
\Cov\qty[Z(i), \int_0^1 X(j) \d j]^2 \le \Var\qty[\int_0^1 X(j) \d j] = \langle \varphi, P\varphi \rangle.
\end{equation}
Combining \eqref{pf_P2} and \eqref{pf_P3}, we obtain \eqref{pf_P1}.
Hence, it follows that
\[
\|P\varphi\|^2 = \int_0^1 \left\| \int_0^1 P(i,j) \varphi(j) \d j \right\|^2 \d i \le \int_0^1 \langle \varphi, P\varphi \rangle \d i = \langle \varphi, P\varphi \rangle,
\]
as desired.
\end{proof}


\subsection{Proof of Theorem~\ref{thm_well}} \label{pf_well}

By Lemma~\ref{lem_affine}, the space of affine strategy profiles can be identified with the Hilbert space $\calL_2^{\bar{d}+1}[0,1]$, and $\bar{\varphi} \in \calL_2^{\bar{d}+1}[0,1]$ constitutes an affine equilibrium if and only if it solves the Fredholm equation \eqref{int_bne}.
It will be convenient to rewrite this equation as
\begin{equation} \label{eq_fredholm}
(I-T)(\bar{\varphi}) = f,
\end{equation}
where $T$ is the linear operator acting on $\calL_2^{\bar{d}+1}[0,1]$ defined by
\begin{equation} \label{def_T}
T\bar{\varphi} = \int_0^1 w(\cdot,j)\bar{P}(\cdot,j)\bar{\varphi}(j) \d j \quad \text{with} \quad \bar{P} (\cdot, j) \coloneqq \mqty[1 & {\bm 0} \\ {\bm 0}^\top & P(\cdot, j)] \in \R^{(\bar{d}+1) \times (\bar{d}+1)},
\end{equation}
where $\bm{0} \in \R^{\bar{d}}$ denotes the column vector of zeros.
Also, $f$ is the function on $[0,1]$ defined by
\begin{equation} \label{def_f}
f(\cdot) = \mqty[\mu_\theta b(\cdot) + c(\cdot) \\ \sigma_\theta b(\cdot) P_\theta(\cdot)] \in \R^{\bar{d}+1}.
\end{equation}
The equation~\eqref{eq_fredholm} suggests that a unique affine equilibrium $\bar{\varphi}$ exists when the operator $I-T$ is invertible.
This idea is formalized by the Fredholm alternative, and to apply it, we first establish several lemmas on basic properties of the objects appearing in \eqref{eq_fredholm}.

\begin{lemma} \label{lem_cts_f}
$f \in \calL_2^{\bar{d}+1}[0,1]$.
\end{lemma}

\begin{proof}
By the definition of $f$, we have
\[
\|f\|^2 = \|\mu_\theta b + c\|^2 + \sigma_\theta^2 \int_0^1 b(i)^2 P_\theta(i)^\top P_\theta(i) \d i \le \|\mu_\theta b + c\|^2 + \sigma_\theta^2 \|b\|^2 < \infty,
\]
where the weak inequality follows from Lemma~\ref{lem_P}~\eqref{lem_P1}.
\end{proof}

\begin{lemma} \label{lem_matrix_CS}
$\|\bar{P}(i,j)\bar{\varphi}(j)\| \le \|\bar{\varphi}(j)\|$ for all $\bar{\varphi} \in \calL_2^{\bar{d}+1}[0,1]$ and $i,j \in [0,1]$.
\end{lemma}

\begin{proof}
By the definition of $\bar{P}$, we have
\[
\|\bar{P}(i,j)\bar{\varphi}(j)\|^2 = |\varphi_0(j)|^2 + \|P(i,j)\varphi(j)\|^2 \le |\varphi_0(j)|^2 + \|\varphi(j)\|^2 = \|\bar{\varphi}(j)\|^2,
\]
where the inequality follows from Lemma~\ref{lem_P}~\eqref{lem_P2}.
\end{proof}

\begin{lemma} \label{lem_HS_norm}
$T:\calL_2^{\bar{d}+1}[0,1] \to \calL_2^{\bar{d}+1}[0,1]$ is a bounded linear operator with $\vertd{T} \le \|w\|$.
\end{lemma}

\begin{proof}
Linearity of $T$ on $\calL_2^{\bar{d}+1}[0,1]$ is immediate.
For boundedness, take any $\bar{\varphi} \in \calL_2^{\bar{d}+1}[0,1]$.
For any $i \in [0,1]$, we have
\begin{align}
\|T\bar{\varphi}(i)\|
&\le \int_0^1 |w(i,j)| \cdot \|\bar{P}(i,j)\bar{\varphi}(j)\| \d j
\le \int_0^1 |w(i,j)| \cdot \|\bar{\varphi}(j)\| \d j \nonumber \\
&\le \qty(\int_0^1 |w(i,j)|^2 \d j)^{\frac{1}{2}} \cdot \qty(\int_0^1 \|\bar{\varphi}(j)\|^2 \d j)^{\frac{1}{2}}
= \qty(\int_0^1 |w(i,j)|^2 \d j)^{\frac{1}{2}} \cdot \|\bar{\varphi}\|, \label{eq_norm_bound_i}
\end{align}
where the first inequality follows from the triangle inequality, the second from Lemma~\ref{lem_matrix_CS}, and the third from the Cauchy--Schwarz inequality.
Squaring both sides and integrating over $i \in [0,1]$, we obtain
\[
\|T\bar{\varphi}\|^2
= \int_0^1 \|T\bar{\varphi}(i)\|^2 \d i
\le \qty(\int_0^1 \int_0^1 |w(i,j)|^2 \d j \d i) \cdot \|\bar{\varphi}\|^2
= \|w\|^2 \cdot \|\bar{\varphi}\|^2.
\]
Thus, $\|T\bar{\varphi}\| \le \|w\| \cdot \|\bar{\varphi}\|$.
Since $\bar{\varphi}$ is arbitrary, we conclude that $\vertd{T} \le \|w\|$.
\end{proof}

The next key lemma strengthens Lemma~\ref{lem_HS_norm}.
In addition to being bounded, the operator $T$ is compact, i.e., it maps any bounded sequence in $\calL_2^{\bar{d}+1}[0,1]$ into a pre-compact subset of $\calL_2^{\bar{d}+1}[0,1]$.
Since its verification requires an additional technical tool, we relegate the proof of Lemma~\ref{lem_compact} to Appendix~\ref{app_misc}.

\begin{lemma} \label{lem_compact}
$T: \calL_2^{\bar{d}+1}[0,1] \to \calL_2^{\bar{d}+1}[0,1]$ is a compact linear operator.
\end{lemma}

Having established that $T$ is a compact operator on the Hilbert space $\calL_2^{\bar{d}+1}[0,1]$, Theorem 3.4 of \cite{krees} implies that equation~\eqref{eq_fredholm} has a unique solution $\bar{\varphi}$ in $\calL_2^{\bar{d}+1}[0,1]$ if and only if the operator $I-T$ is injective.
At the same time, $I-T$ is injective if and only if there exists no nonzero element $\bar{\psi} \in \calL_2^{\bar{d}+1}[0,1]$ such that $\bar{\psi}=T\bar{\psi}$, which is equivalent to $1$ not being an eigenvalue of $T$.
These observations lead to the next lemma, where we write $T$ as $T_{w,P}$ to emphasize its dependence on primitives.

\begin{lemma} \label{lem_alt}
An LQG game $\calG = (w,b,c,\calP)$ has a unique affine equilibrium if and only if $I-T_{w,P}$ has a bounded inverse, which is further equivalent to $1 \notin \Lambda (T_{w,P})$.
\end{lemma}

Any compact operator has the important property that its set of eigenvalues has no accumulation point, except possibly $0$ \citep[Theorem 3.9 of][]{krees}.
Hence, if $1 \in \Lambda (T_{w,P})$, then we can take $\delta > 0$ such that $[1,1+\delta] \cap \Lambda (T_{w,P}) = \{1\}$.
Given any small $\epsilon > 0$, multiplying both sides by $(1-\epsilon)$ gives
\[
\qty[1-\epsilon,\, (1-\epsilon)(1+\delta)] \cap (1-\epsilon)\Lambda (T_{w,P}) = \{1-\epsilon\},
\]
while the construction of $T_{w,P}$ and the homogeneity of eigenvalues imply
\[
(1-\epsilon) \Lambda \qty(T_{w,P}) = \Lambda \qty((1-\epsilon)T_{w,P}) = \Lambda (T_{(1-\epsilon)w,\calP}).
\]
Moreover, $1 \in \qty[1-\epsilon,\, (1-\epsilon)(1+\delta)]$ whenever $\epsilon \le \delta/(1+\delta)$.
Thus, the LQG game $\calG_\epsilon$ obtained from $\calG$ by replacing $w$ with $(1-\epsilon)w$ has a unique affine equilibrium whenever $\epsilon \in (0,\delta/(1+\delta)]$.
Since the distance between $\calG$ and $\calG_\epsilon$ can be made arbitrarily small, this shows that unique existence is a ``dense'' property.

Next, we show that unique existence is an ``open'' property.
To this end, consider any sequence of LQG games $\calG_n = (w_n,b_n,c_n,\calP_n)$ that converges to the limit LQG game $\calG = (w,b,c,\calP)$, and assume that $\calG$ has a unique affine equilibrium.
Let $T_n$ and $f_n$ be defined as in \eqref{def_T} and \eqref{def_f}, using the primitives of $\calG_n$.
We use $T$ and $f$ for the corresponding objects associated with $\calG$.
By Lemma~\ref{lem_alt}, we have $1 \notin \Lambda (T)$.

\begin{lemma} \label{lem_cts_op}
$\|f-f_n\| \to 0$ and $\vertd{T - T_n} \to 0$ as $n \to \infty$.
\end{lemma}

\begin{proof}
We first show that $b_nP_{\theta,n}:[0,1] \to \R^{\bar{d}}$ converges to $bP_\theta:[0,1] \to \R^{\bar{d}}$ in $L_2$.
The triangle inequality yields
\[
\|bP_\theta - b_nP_{\theta,n}\| \le \|b(P_\theta - P_{\theta,n})\| + \|(b-b_n)P_{\theta,n}\|.
\]
By Lemma~\ref{lem_P}~\eqref{lem_P1}, the second term is bounded by $\|b-b_n\|$, which converges to $0$ as $n \to \infty$.
Moreover, since $\|P_\theta(i)-P_{\theta,n}(i)\|$ is uniformly bounded by $2$ by Lemma~\ref{lem_P}~\eqref{lem_P1}, the dominated convergence theorem implies
\[
\lim_{n \to \infty} \|b(P_\theta - P_{\theta,n})\|^2 = \int_0^1 |b(i)|^2 \cdot \qty(\lim_{n \to \infty} \|P_\theta(i) - P_{\theta,n}(i)\|^2) \d i = 0.
\]
Thus, $\|bP_\theta-b_nP_{\theta,n}\| \to 0$ as $n \to \infty$, which in turn implies that
\[
\lim_{n\to \infty} \|f-f_n\|^2 = \lim_{n\to \infty} \qty(\|\mu_\theta (b-b_n) + (c-c_n)\|^2 + \sigma^2_\theta \|bP_\theta - b_n P_{\theta,n} \|^2) = 0.
\]

To show that $\vertd{T - T_n} \to 0$, take any $\bar{\varphi} \in \calL_2^{\bar{d}+1}[0,1]$.
For any $i \in [0,1]$, calculations analogous to \eqref{eq_norm_bound_i} give
\[
\left\|(T-T_n)(\bar{\varphi})[i] \right\|
\le \qty(\int_0^1 \vertb{w(i,j)\bar{P}(i,j) - w_n(i,j)\bar{P}_n(i,j)}^2 \d j)^{\frac{1}{2}} \cdot \|\bar{\varphi}\|.
\]
Since $\bar{\varphi}$ is arbitrary, squaring both sides and integrating over $i \in [0,1]$, we obtain the following upper bound for the operator norm:
\begin{equation} \label{eq_T_conv1}
\vertd{T-T_n}^2 \le \int_0^1 \int_0^1 \vertb{w(i,j)\bar{P}(i,j) - w_n(i,j)\bar{P}_n(i,j)}^2 \d j \d i.
\end{equation}
By Lemma~\ref{lem_matrix_CS}, $\vertt{\bar{P}(i,j)} \le 1$ and $\vertt{\bar{P}_n(i,j)} \le 1$ for all $i,j \in [0,1]$ and $n \in \N$.
Together with the triangle inequality, it holds for all $i,j \in [0,1]$ that
\[
\vertb{w(i,j)\bar{P}(i,j) - w_n(i,j)\bar{P}_n(i,j)} \le |w(i,j)| \cdot \vertb{\bar{P}(i,j) - \bar{P}_n(i,j)} + \qty|w(i,j)-w_n(i,j)|.
\]
Squaring both sides and substituting into \eqref{eq_T_conv1}, it follows that
\begin{align*}
\vertd{T-T_n}^2 \le
& \underbrace{\int_0^1 \int_0^1 \qty|w(i,j) - w_n(i,j)|^2 \d j \d i}_{\rm (a)} \\
& \quad + \underbrace{\int_0^1 \int_0^1 \qty|w(i,j)|^2 \cdot \vertb{\bar{P}(i,j) - \bar{P}_n(i,j)}^2 \d j \d i}_{\rm (b)} \\
& \quad + \underbrace{2 \int_0^1 \int_0^1 \qty|w(i,j)| \cdot \qty|w(i,j)-w_n(i,j)| \cdot \vertb{\bar{P}(i,j) - \bar{P}_n(i,j)} \d j \d i}_{\rm (c)}.
\end{align*}
By assumption, ${\rm (a)}$ converges to $0$ as $n \to \infty$.
Also, Lemma~\ref{lem_matrix_CS} and the Cauchy--Schwarz inequality imply that
\begin{align*}
{\rm (c)} 
&\le 4 \int_0^1 \int_0^1 \qty|w(i,j)| \cdot \qty|w(i,j)-w_n(i,j)| \d j \d i \\
&\le 4 \qty(\int_0^1 \int_0^1 \qty|w(i,j)|^2 \d j \d i)^\frac{1}{2} \cdot \qty(\int_0^1 \int_0^1 \qty|w(i,j)-w_n(i,j)|^2 \d j \d i)^\frac{1}{2}
= 4 \|w\| \cdot \|w-w_n\|,
\end{align*}
where the right-hand side converges to $0$ as $n \to \infty$.
As for ${\rm (b)}$, since $P_n \to P$ in $L_2$ and $\vertb{\bar{P}(i,j)-\bar{P}_n(i,j)}$ is uniformly bounded by $2$, the dominated convergence theorem yields
\begin{align*}
\lim_{n\to \infty} &\int_0^1 \int_0^1 \qty|w(i,j)|^2 \cdot \vertb{\bar{P}(i,j) - \bar{P}_n(i,j)}^2 \d j \d i \\
= &\int_0^1 \int_0^1 \qty|w(i,j)|^2 \cdot \qty( \lim_{n \to \infty} \vertb{\bar{P}(i,j) - \bar{P}_n(i,j)}^2) \d j \d i = 0.
\end{align*}
Putting these estimates together, we conclude that $\vertd{T-T_n} \to 0$ as $n \to \infty$.
\end{proof}

Now, suppose by contradiction that there is no unique affine equilibrium in any of $\calG_n$.
Then, by Lemma~\ref{lem_alt}, we have $1 \in \Lambda (T_n)$ for all $n \in \N$, while $1 \notin \Lambda (T)$ by assumption.
Thus, for each $n$, there exists a nonzero element $\bar{\psi}^n \in \calL_2^{\bar{d}+1}[0,1]$ such that $(I-T_n)(\bar{\psi}^n) = 0$.
By the linearity of eigenvectors, we can normalize each $\bar{\psi}^n$ so that $\|\bar{\psi}^n\| = 1$, yielding a bounded sequence $\{\bar{\psi}^n\}_{n \in \N}$.
The next lemma shows that this sequence has a convergent subsequence.
Its proof is handled together with the proof of Lemma~\ref{lem_compact}, and is therefore relegated to Appendix~\ref{app_misc}.

\begin{lemma} \label{lem_psi_conv}
$\{\bar{\psi}^n\}_{n \in \N}$ contains a convergent subsequence with limit $\bar{\psi} \in \calL_2^{\bar{d}+1}[0,1]$.
\end{lemma}

Abusing notation, let $\{\bar{\psi}^n\}_{n \in \N}$ denote the convergent subsequence found in Lemma~\ref{lem_psi_conv}, and let $\bar{\psi} \in \calL_2^{\bar{d}+1}[0,1]$ be its limit.
Since $\|\bar{\psi}^{n}\| = 1$ for all $n$, we have $\|\bar{\psi}\| = 1$, so $\bar{\psi}$ is nonzero.
Moreover, the triangle inequality yields
\begin{align*}
&\|(I-T)(\bar{\psi}) - (I-T_n)(\bar{\psi}^n) \| \nonumber \\
& \quad \le \| \bar{\psi} - \bar{\psi}^n \| + \| T\bar{\psi} - T_n\bar{\psi}^n \| \nonumber \\
& \quad \le \| \bar{\psi} - \bar{\psi}^n \| + \| T (\bar{\psi} - \bar{\psi}^n) \| + \| \qty(T - T_n) \bar{\psi}^n \| \nonumber \\
& \quad \le \| \bar{\psi} - \bar{\psi}^n \| + \vertd{T} \cdot \|\bar{\psi} - \bar{\psi}^n \| + \vertd{T - T_n} \cdot \|\bar{\psi}^n \|,
\end{align*}
where the last line converges to $0$ as $n \to \infty$ by Lemma~\ref{lem_cts_op}.
Hence, it follows that
\begin{align*}
\|(I-T)(\bar{\psi})\| \le \underbrace{\|(I-T)(\bar{\psi}) - (I-T_n)(\bar{\psi}^n) \|}_{\to \hspace{1pt} 0 \text{ as } n \hspace{1pt} \to \hspace{1pt} \infty} + \underbrace{\|(I-T_n)(\bar{\psi}^n) \|}_{= \hspace{1pt} 0 \text{ by assumption}},
\end{align*}
which implies that $\bar{\psi} = T(\bar{\psi})$.
However, since $\bar{\psi}$ is nonzero, this implies $1 \in \Lambda (T)$, a contradiction.
Therefore, we have shown that whenever $\calG_n \to \calG$ and $\calG$ has a unique affine equilibrium, $\calG_n$ also has a unique affine equilibrium for every sufficiently large $n$.

Lastly, we show that affine equilibria converge whenever the corresponding sequence of primitives converges.
Let $\bar{\varphi}$ be the unique affine equilibrium of $\calG$, and let $\bar{\varphi}^n$ be the unique affine equilibrium of $\calG_n$, where $\calG_n \to \calG$.
By Lemma~\ref{lem_alt}, $I-T_n$ has a bounded inverse for every $n \in \N$.
The next lemma shows that these inverse operators are uniformly bounded.

\begin{lemma} \label{lem_cts_inv}
It holds that
\begin{equation} \label{bound_inv}
\sup_{n \in \N} \vertb{(I-T_n)^{-1}} < \infty.
\end{equation}
\end{lemma}

\begin{proof}
Since $I-T$ has a bounded inverse, we have $\vertt{(I-T)^{-1}} < \infty$.
Moreover, since $T_n$ converges to $T$ by Lemma~\ref{lem_cts_op}, there exists $n_0 \in \N$ such that 
\begin{equation} \label{eq_large_n_T}
\vertd{T-T_n} < \frac{1}{2\vertd{(I-T)^{-1}}}, \quad \forall n \ge n_0.
\end{equation}
For an arbitrary $n \in \N$, the triangle inequality yields
\begin{align*}
\vertb{(I-T_n)^{-1}}
&\le \vertb{(I-T_n)^{-1} - (I-T)^{-1}} + \vertb{(I-T)^{-1}} \\
&= \vertb{(I-T_n)^{-1}(T-T_n)(I-T)^{-1}} + \vertb{(I-T)^{-1}} \\
&\le \vertb{(I-T_n)^{-1}} \cdot \vertb{T-T_n} \cdot \vertb{(I-T)^{-1}} + \vertb{(I-T)^{-1}},
\end{align*}
where the second line follows from the resolvent identity \citep[see, e.g., equation (24) on p.\ 198 of][]{lax}, and the third line follows from the definition of the operator norm.
Now, if $n \ge n_0$ so that \eqref{eq_large_n_T} holds, we have
\[
\vertb{(I-T_n)^{-1}} \le \frac{\vertd{(I-T)^{-1}}}{1-\vertd{(I-T)^{-1}} \cdot \vertd{T-T_n}} < \frac{\vertd{(I-T)^{-1}}}{1-\frac{1}{2}} = 2 \vertb{(I-T)^{-1}},
\]
which gives a uniform bound for $\vertt{(I-T_n)^{-1}}$ for all $n \ge n_0$.
Since $n_0$ is finite and $I-T_n$ has a bounded inverse for each $n < n_0$, it follows that $\sup_{n \in \N} \vertt{(I-T_n)^{-1}} < \infty$.
\end{proof}

By \eqref{eq_fredholm}, the unique affine equilibria are given by $\bar{\varphi} = (I-T)^{-1}(f)$ and $\bar{\varphi}^n = (I-T_n)^{-1}(f_n)$ for each $n \in \N$.
Then,
\begin{align}
&\|\bar{\varphi} - \bar{\varphi}^n\|
= \bigl\|(I-T)^{-1}(f) - (I-T_n)^{-1}(f_n) \bigr\| \nonumber \\
& \quad \le \bigl\|(I-T)^{-1}(f-f_n) \bigr\| + \bigl\| \bigl[(I-T)^{-1}-(I-T_n)^{-1} \bigr](f_n) \bigr\| \nonumber \\
& \quad = \bigl\|(I-T)^{-1}(f-f_n) \bigr\| + \bigl\| \bigl[ (I-T)^{-1}(T-T_n)(I-T_n)^{-1} \bigr](f_n) \bigr\| \nonumber \\
& \quad \le \vertb{(I-T)^{-1}} \cdot \|f-f_n\| + \vertb{(I-T)^{-1}} \cdot \vertb{T-T_n} \cdot \vertb{(I-T_n)^{-1}} \cdot \|f_n\| \nonumber \\
& \quad \le \vertb{(I-T)^{-1}} \cdot \biggl(\|f-f_n\| + \underbrace{\biggl(\sup_{n \in \N}\vertb{(I-T_n)^{-1}} \cdot \sup_{n \in \N} \|f_n\| \biggr)}_{(*)} \cdot \vertb{T-T_n}  \biggr), \label{eq_affine_conv}
\end{align}
where the second line follows from the triangle inequality, the third from the resolvent identity, and the fourth from the definition of the operator norm.
In the last line of \eqref{eq_affine_conv}, the term $(*)$ is finite because $f_n$ is convergent and Lemma~\ref{lem_cts_inv} gives the uniform boundedness of $\vertb{(I-T_n)^{-1}}$.
Lemma~\ref{lem_cts_op} then implies that $\|f-f_n\|$ and $\vertt{T-T_n}$ converge to $0$ as $n \to \infty$.
Therefore, $\|\bar{\varphi} - \bar{\varphi}^n\| \to 0$ as $n \to \infty$.
\hfill {\it Q.E.D.}


\subsection{Proof of Theorem \ref{thm_equiv}}

\begin{proof}[Proof of Lemma~\ref{lem_canonical_psd}]
Part~\eqref{canonical_psd1} is equivalent to the normalization condition $\Var\qty[S^*(i)] = 1$.
Moreover, as per \eqref{P_cano}, the covariance matrix \eqref{cov_matrix_join} for any canonical information structure is given as
\[
\mqty[{\bm K}(N) & {\bm K}_\theta(N) \\ {\bm K}_\theta(N)^\top & \sigma^2_\theta] 
= \mqty[\sigma^2_\theta h(i_1)^2 + g(i_1,i_1) & \cdots & \sigma^2_\theta h(i_1)h(i_n) + g(i_1,i_n) & \sigma^2_\theta h(i_1) \\
\vdots & \ddots & \vdots & \vdots \\
\sigma^2_\theta h(i_n)h(i_1) + g(i_1, i_n) & \cdots & \sigma^2_\theta h(i_n)^2 + g(i_n,i_n) & \sigma^2_\theta h(i_n) \\
\sigma^2_\theta h(i_1) & \cdots & \sigma^2_\theta h(i_n) & \sigma^2_\theta].
\]
Clearly, this matrix is symmetric if and only if Part~\eqref{canonical_psd2} is satisfied.
Moreover, since $\sigma^2_\theta > 0$, the matrix is positive semidefinite if and only if the following Schur complement is positive semidefinite:
\begin{align*}
\mqty[\sigma^2_\theta h(i_1)^2 + g(i_1,i_1) & \cdots & \sigma^2_\theta h(i_1)h(i_n) + g(i_1,i_n) \\
\vdots & \ddots & \vdots \\
\sigma^2_\theta h(i_n)h(i_1) + g(i_1, i_n) & \cdots & \sigma^2_\theta h(i_n)^2 + g(i_n,i_n)]
- \frac{1}{\sigma^2_\theta} \mqty[\sigma^2_\theta h(i_1) \\ \vdots \\\sigma^2_\theta h(i_1)] \mqty[\sigma^2_\theta h(i_1) \\ \vdots \\\sigma^2_\theta h(i_1)]^\top,
\end{align*}
which reduces to \eqref{g_psd}, and thus the condition is precisely Part~\eqref{canonical_psd3}.
\end{proof}

Let us prove Theorem~\ref{thm_equiv}.
Let $\calX$ be a Gaussian BCE under a payoff structure $(w,b,c)$, so that its first and second moments satisfy \eqref{bce1} and \eqref{bce2}.
Define $f: [0,1] \to \{+1,0,-1\}$ by $f(i) = \sgn(\Cov[X(i), \theta])$.
Also, let $I=\{i \in [0,1]: \Var\qty[X(i)] > 0\}$ be the set of agents whose BCE strategies are not deterministic.
Note that the correlation coefficients
\[
\Corr\qty[X(i), \theta]=\frac{\Cov\qty[X(i),\theta]}{\sigma_\theta \Var\qty[X(i)]^{\frac{1}{2}}}, \quad \Corr\qty[X(i), X(j)]=\frac{\Cov\qty[X(i),X(j)]}{\Var\qty[X(i)]^{\frac{1}{2}}\Var\qty[X(j)]^{\frac{1}{2}}}
\]
are well-defined for any $i,j \in I$.

Now, let $\calP^*=(h,g)$ be a canonical information structure constructed as follows:
\begin{itemize}
\item for $i \notin I$, set $g(i,i)=1$ and $h(i)=g(i,j)=0$ for all $j \neq i$;
\item for $i \in I$, set $h(i) = f(i)\Corr\qty[X(i), \theta] / \sigma_\theta$ and
\[
g(i,j) = \begin{cases}
f(i)f(j)\qty(\Corr\qty[X(i),X(j)] - \Corr\qty[X(i), \theta] \Corr\qty[X(j),\theta]) &\text{for } j \in I, \\
0 &\text{for } j \notin I.
\end{cases}
\]
\end{itemize}
Note that $h(i) \ge 0$ for all $i \in [0,1]$ by construction.
Let us show that $\calP^*$ is well-defined in light of Lemma~\ref{lem_canonical_psd}.
Noticing that $f(i)^2=1$ for any $i \in I$, the properties \eqref{canonical_psd1} and \eqref{canonical_psd2} follow directly from construction.
To verify \eqref{canonical_psd3}, let $G$ denote the matrix given as in \eqref{g_psd} for $\{i_1,\ldots,i_n\}$.
It is without loss of generality to assume $\{i_1,\ldots,i_n\} \subseteq I$; otherwise, if $\Var\qty[X(i_k)] = 0$ for some $i_k$, the matrix \eqref{g_psd} is positive semidefinite if and only if its principal minor obtained by removing the $i_k$-th row and column is positive semidefinite.
Given this nonzero variance condition, consider the normalized variables $Z_0 = \frac{\theta - \mu_\theta}{\sigma_\theta}$ and $Z_k = \frac{X(i_k)-\E\qty[X(i_k)]}{\Var\qty[X(i_k)]^{1/2}}$ for each $i_k$.
In addition, define random variables
\[
\tilde{Z}_k = Z_k - \E\qty[Z_0Z_k] Z_0, \quad \forall k=1,\ldots,n.
\]
Since $\E|Z_0|^2 = 1$, it follows that
\begin{align*}
\E\qty[\tilde{Z}_k \tilde{Z}_l]
&= \E\qty[Z_kZ_l]-\E\qty[Z_0Z_k]\E\qty[Z_0Z_l] \\
&= \Corr\qty[X(i),X(j)] - \Corr\qty[X(i), \theta] \Corr\qty[X(j),\theta]
= f(i_k)f(i_l)g(i_k,i_l).
\end{align*}
This implies that the matrix
\[
\tilde{G} = \mqty[f(i_1)^2 g(i_1,i_1) & \cdots & f(i_1)f(i_n)g(i_1,i_n) \\
\vdots & \ddots & \vdots \\
f(i_n)f(i_1) g(i_n,i_1) & \cdots & f(i_n)^2g(i_n,i_n)]
\]
 is obtained as the covariance matrix of $(\tilde{Z}_1,\ldots,\tilde{Z}_n)$, and is therefore positive semidefinite.
 Moreover, since $\tilde{G} = FGF$ holds with $F=\Diag(f(i_1),\ldots,f(i_n))$ being a nonsingular diagonal matrix, it follows that $G$ is positive semidefinite.

Next, we show that an affine equilibrium $(\varphi^*_0,\varphi_1^*)$ in the game $\calG^* = (w,b,c,\calP^*)$ is given by
\[
\varphi_0^*(i)=\E\qty[X(i)] \quad \text{and} \quad
\varphi_1^*(i)=f(i)\Var\qty[X(i)]^{\frac{1}{2}}, \quad \forall i \in [0,1].
\]
Substituting $\varphi_0^*(i)=\E\qty[X(i)]$ into \eqref{bce1} directly verifies that $\varphi_0^*$ solves the first line of \eqref{int_bne}.
Moreover, since \eqref{P_cano} implies that $P(i,j)=P_\theta(i)=0$ for $i \notin I$ and $j \neq i$, the second line of \eqref{int_bne} is satisfied for any $i \notin I$.
For any $i \in I$, again by \eqref{P_cano}, the right-hand side of the second line of \eqref{int_bne} can be arranged as
\begin{align*}
&\int_0^1 w(i,j)P(i,j)\varphi_1^*(j) \d j+\sigma_\theta b(i) P_\theta(i) \\
&= \int_0^1 w(i,j)\qty(\sigma^2_\theta h(i)h(j) + g(i,j))\varphi_1^*(j) \d j+\sigma^2_\theta b(i) h(i)  \\
&= \int_I w(i,j) \cdot f(i)f(j) \Corr\qty[X(i),X(j)] \cdot f(j)\Var\qty[X(j)]^{\frac{1}{2}} \d j + \sigma_\theta b(i) f(i) \Corr\qty[X(i),\theta] \\
&= \frac{f(i)}{\Var\qty[X(i)]^{\frac{1}{2}}} \cdot \underbrace{\qty(\int_I w(i,j) \Cov \qty[X(i), X(j)] \d j + b(i) \Cov \qty[X(i), \theta])}_{= \Var\qty[X(i)] \text{ by \eqref{bce2}}}
= f(i)\Var\qty[X(i)]^{\frac{1}{2}}
= \varphi_1^*(i).
\end{align*}
Thus, $\varphi_1^*$ solves the second line of \eqref{int_bne}, and so $(\varphi^*_0,\varphi_1^*)$ forms an affine equilibrium in $\calG^*$.

Lastly, let us show that the induced outcome $\calX^*$ of $(\bar{\varphi}^*, \calG^*)$ coincides with $\calX$.
As for the first moments, by Lemma~\ref{lem_affine}, $\E\qty[X^*(i)] = \varphi_0^*(i) = \E\qty[X(i)]$.
As for the second moments, if $i \notin I$, then $\Var\qty[X(i)]=0$ and $\varphi_1^*(i)=0$.
Thus, $\Cov\qty[X^*(i), \theta] = \Cov\qty[X(i), \theta] = 0$ and $\Cov\qty[X^*(i), X^*(j)] = \Cov\qty[X(i), X(j)] = 0$ for any $j \neq i$.
Moreover, for $i \in I$, since $f(i)^2 = 1$,
\begin{align*}
\Cov\qty[X^*(i), \theta] &= \sigma_\theta \varphi^*_1(i) P_\theta(i) = \sigma^2_\theta h(i) \varphi^*_1(i) \\
&= \sigma^2_\theta \cdot \frac{f(i)\Corr\qty[X(i), \theta]}{\sigma_\theta} \cdot f(i)\Var\qty[X(i)]^{\frac{1}{2}} = \Cov\qty[X(i), \theta].
\end{align*}
Also, for $i,j \in I$,
\begin{align*}
&\Cov\qty[X^*(i), X^*(j)] = \varphi^*_1(i) \varphi^*_1(j) P(i,j)
= \qty(\sigma^2_\theta h(i)h(j) + g(i,j)) \varphi^*_1(i) \varphi^*_1(j) \\
&\quad = f(i)f(j)\Corr\qty[X(i),X(j)] \cdot f(i)\Var\qty[X(i)]^{\frac{1}{2}} \cdot f(j)\Var\qty[X(j)]^{\frac{1}{2}}
= \Cov\qty[X(i), X(j)].
\end{align*}
Therefore, $\calX^*$ and $\calX$ have the same first and second moments.
Since both are Gaussian processes, they have the same distribution.
\hfill {\it Q.E.D.}


\subsection{Proof of Theorem \ref{thm_id}}

Let $\calX$ be a given Gaussian BCE under some unknown payoff structure $(w,b,c)$.
Theorem~\ref{thm_equiv} guarantees that there exists some canonical information structure $\calP^* = (h,g)$ and an affine equilibrium $(\varphi^*_0,\varphi_1^*)$ in the game $\calG^* = (w,b,c,\calP^*)$ such that $\calX$ satisfies \eqref{eq_id}.
We want to show that the conditional action distributions $X(\cdot)|_{\theta = \hat{\theta}_t}$ for two distinct values $\hat{\theta}_1, \hat{\theta}_2 \in \R$ uniquely pin down $(h,g,\varphi_0^*,\varphi_1^*)$.

First, take the conditional expectation of $X(i)$ given each of $\hat{\theta}_1$ and $\hat{\theta}_2$.
Since $\varepsilon(i)$ is independent of $\theta$, we have
\begin{align*}
\hat{X}_1(i) = \E\qty[X(i) \mid \theta = \hat{\theta}_1] = \varphi_0^*(i) + \varphi_1^*(i)h(i) \qty(\hat{\theta}_1 - \mu_\theta), \\
\hat{X}_2(i) = \E\qty[X(i) \mid \theta = \hat{\theta}_2] = \varphi_0^*(i) + \varphi_1^*(i)h(i) \qty(\hat{\theta}_2 - \mu_\theta).
\end{align*}
Taking the difference between these equations yields
\begin{equation} \label{pf_beta_id}
\varphi_1^*(i)h(i) = \hat{\beta}(i) \equiv \frac{\hat{X}_1(i) - \hat{X}_2(i)}{\hat{\theta}_1 - \hat{\theta}_2}.
\end{equation}
Substituting this back into the expression for $\hat{X}_1(i)$, or equivalently for $\hat{X}_2(i)$, identifies $\varphi_0^*(i)$ as $\varphi_0^*(i) = \hat{X}_1(i) - \hat{\beta}(i)(\hat{\theta}_1 - \mu_\theta)$.

In what follows, suppose that Assumption~\ref{asm_var} is satisfied.
Taking the conditional covariance between $X(i)$ and $X(j)$ given $\theta = \hat{\theta}_1$, we have
\begin{align}
\hat{\gamma}(i,j)
&= \Cov\qty[X(i), X(j) \mid \theta = \hat{\theta}_1] 
= \E \qty[\qty(X(i)-\hat{X}_1(i))\qty(X(j)-\hat{X}_1(j)) \mid \theta = \hat{\theta}_1] \nonumber \\
&= \varphi_1^*(i) \varphi_1^*(j) \E \qty[\varepsilon(i)\varepsilon(j) \mid \theta = \hat{\theta}_1] = \varphi_1^*(i) \varphi_1^*(j) g(i,j). \label{pf_gamma_id1}
\end{align}
In particular, the conditional action variance of $X(i)$ yields
\begin{equation} \label{pf_gamma_id2}
\hat{\gamma}(i,i) = \Var\qty[X(i) \mid \theta = \hat{\theta}_1] = \varphi_1^*(i)^2 g(i,i).
\end{equation}
In addition, by the normalization of canonical signal variance, Lemma~\ref{lem_canonical_psd} implies that $h$ and $g$ satisfy
\begin{equation} \label{pf_norm}
\sigma^2_\theta h(i)^2 + g(i,i) = 1.
\end{equation}
Multiplying both sides of \eqref{pf_norm} by $\varphi_1^*(i)^2$ and substituting \eqref{pf_beta_id} and \eqref{pf_gamma_id2}, we get
\[
\varphi_1^*(i)^2 = \sigma^2_\theta \qty(\varphi_1^*(i)h(i))^2 + \varphi_1^*(i)^2 g(i,i)
= \sigma^2_\theta \hat{\beta}(i)^2 + \hat{\gamma}(i,i),
\]
from which $|\varphi_1^*(i)| = \sqrt{\sigma^2_\theta \hat{\beta}(i)^2 + \hat{\gamma}(i,i)}$ is identified.
Note that Assumption~\ref{asm_var} ensures that $\sigma^2_\theta \hat{\beta}(i) + \hat{\gamma}(i,i) > 0$ holds by the law of total variance \eqref{law_var}, so $\varphi_1^*(i) \neq 0$.
Then, substituting this into the square of \eqref{pf_beta_id}, and using the nonnegativity $h(i)\ge 0$, identifies $h(i)$ as
\[
h(i) = \frac{|\hat{\beta}(i)|}{\sqrt{\sigma^2_\theta \hat{\beta}(i)^2 + \hat{\gamma}(i,i)}}.
\]
This identified value of $h(i)$ in turn determines $g(i,i)$ through \eqref{pf_norm} as
\[
g(i,i) = 1-\sigma^2_\theta h(i)^2 = \frac{\hat{\gamma}(i,i)}{\sigma^2_\theta \hat{\beta}(i)^2 + \hat{\gamma}(i,i)}.
\]
In addition, the absolute value of $g(i,j)$ for $i \neq j$ is identified from \eqref{pf_gamma_id1} as
\[
|g(i,j)| = \frac{|\hat{\gamma}(i,j)|}{|\varphi_1^*(i)| \cdot |\varphi_1^*(j)|} = \frac{\hat{\gamma}(i,j)}{\sqrt{(\sigma^2_\theta \hat{\beta}(i)^2 + \hat{\gamma}(i,i))(\sigma^2_\theta \hat{\beta}(j)^2 + \hat{\gamma}(j,j))}}.
\]

Lastly, suppose that Assumption~\ref{asm_beta} is satisfied.
This implies that $\hat{\beta}(i) \neq 0$, and thus $h(i) > 0$.
Hence, $\varphi_1^*(i)$ is identified from \eqref{pf_beta_id} as
\[
\varphi_1^*(i) = \frac{\hat{\beta}(i)}{|\hat{\beta}(i)|} \sqrt{\sigma^2_\theta \hat{\beta}(i)^2 + \hat{\gamma}(i,i)}
= \sgn(\hat{\beta}(i)) \sqrt{\sigma^2_\theta \hat{\beta}(i)^2 + \hat{\gamma}(i,i)}.
\]
Then, noticing that $1/\sgn(\hat{\beta}(i)) = \sgn(\hat{\beta}(i))$ and $\sgn(\hat{\beta}(i)) \sgn(\hat{\beta}(j)) = \sgn(\hat{\beta}(i)\hat{\beta}(j))$, we can identify $g(i,j)$ from \eqref{pf_gamma_id1} as
\[
g(i,j) = \frac{\hat{\gamma}(i,j)}{\varphi_1^*(i) \varphi_1^*(j)} = \frac{\sgn(\hat{\beta}(i)\hat{\beta}(j))\hat{\gamma}(i,j)}{\sqrt{(\sigma^2_\theta \hat{\beta}(i)^2 + \hat{\gamma}(i,i))(\sigma^2_\theta \hat{\beta}(j)^2 + \hat{\gamma}(j,j))}},
\]
which concludes the proof.
\hfill {\it Q.E.D.}


\subsection{Proofs of Propositions~\ref{prop_var}--\ref{prop_actvar}}

\begin{proof}[Proof of Lemma~\ref{lem_var}]
Let $\calX$ be a common induced outcome that arises in both $(w,b,c,\calP)$ and $(w,b,c,\calP^*)$, and let $Z$ be any random variable with $\Var\qty[Z] < \infty$.
Suppose that $\Var\qty[X(i)] > 0$.
Let $\hat{Z} = \E\qty[Z \mid S(i)]$ be the best estimate of $Z$ given agent $i$'s signal $S(i)$ under $\calP$.
Applying the law of total variance to $\hat{Z}$, with $X(i)$ as the conditioning random variable, gives
\[
\Var\qty[\hat{Z}] =\Var\qty[\E\qty[\hat{Z} \mid X(i)]] + \E\qty[\Var\qty[\hat{Z} \mid X(i)]]
\]
The second term is nonnegative, and $\E[\hat{Z} \mid X(i)] = \E[Z \mid X(i)]$ since $X(i)$ is measurable with respect to $S(i)$.
Hence, it follows that
\begin{equation} \label{pf_var_gap}
\Var\qty[\E\qty[Z \mid S(i)]] \ge \Var\qty[\E[Z \mid X(i)]].
\end{equation}
Since $\Var\qty[X(i)] > 0$, we must have $\varphi_1^*(i) > 0$ in the affine equilibrium of $(w,b,c,\calP^*)$, which in turn implies that $S^*(i)$ and $X(i)$ are mutually measurable by \eqref{eq_id}.
Therefore, $\Var\qty[\E[Z \mid X(i)]] = \Var\qty[\E[Z \mid S^*(i)]]$, and hence \eqref{pf_var_gap} yields $\Var\qty[\E\qty[Z \mid S(i)]] \ge \Var\qty[\E[Z \mid S^*(i)]]$.
By the law of total variance, this implies $r_i(Z \mid \calP) \ge r_i(Z \mid \calP^*)$.
\end{proof}


\begin{proof}[Proof of Proposition~\ref{prop_var}]
Let $\calX$ be an induced outcome of $(\calG,\bar{\varphi})$.
As demonstrated in the main text, we have $r_i(\theta \mid \calP) = \|P_\theta(i)\|^2$, while under the identified canonical information structure this expression reduces to $\sigma^2_\theta h(i)^2$.
Thus, Lemma~\ref{lem_var} implies $\|P_\theta(i)\|^2 \ge \sigma^2_\theta h(i)^2$, and substituting the identified formula for $h(i)$ from Theorem~\ref{thm_id} yields \eqref{id_var1}.

As for \eqref{id_var2}, since $X(j) = \varphi_0(j) + \varphi(j)^\top K(j,j)^{-\frac{1}{2}} (S(j) - m(j))$, the conditional Gaussian formula implies
\begin{align*}
&\E\qty[X(j) \mid S(i)] - \varphi_0(j) = \varphi(j)^\top K(j,j)^{-\frac{1}{2}} \qty(\E\qty[S(j) \mid S(i)] - m(j)) \\
&\quad = \varphi(j)^\top K(j,j)^{-\frac{1}{2}} K(j,i) K(i,i)^{-1} \qty(S(i) - m(i))
= \varphi(j)^\top P(j,i) K(i,i)^{-\frac{1}{2}}\qty(S(i)-m(i)).
\end{align*}
Taking the variance of this best estimate of agent $j$'s action, we obtain
\begin{align*}
r_i(X(j) \mid S(i))
&= \frac{\Var\qty[X(j)] - \Var\qty[X(j) \mid S(i)]}{\Var\qty[X(j)]}
= \frac{\Var\qty[\E\qty[X(j) \mid S(i)]]}{\Var\qty[X(j)]} \\
&= \frac{\varphi(j)^\top P(j,i) K(i,i)^{-\frac{1}{2}} K(i,i) K(i,i)^{-\frac{1}{2}}P(i,j) \varphi(j)}{\Var\qty[X(j)]}
= \frac{\|P(i,j)\varphi(j)\|^2}{\Var\qty[X(j)]}.
\end{align*}
In particular, recalling that $P(i,j)$ corresponds to $\sigma^2_\theta h(i)h(j) + g(i,j)$ under the canonical information structure, Lemma~\ref{lem_var} implies
\[
\|P(i,j)\varphi(j)\|^2 \ge \qty(\varphi_1^*(j) \cdot \qty(\sigma^2_\theta h(i)h(j) + g(i,j)))^2.
\]
The right-hand side can be written in terms of the identified objects from Theorem~\ref{thm_id} as
\begin{align*}
\qty(\sigma^2_\theta \hat{\beta}(j)^2 + \hat{\gamma}(j,j)) \cdot \qty(\frac{\sigma^2_\theta |\hat{\beta}(i)||\hat{\beta}(j)| + \sgn(\hat{\beta}(i) \hat{\beta} (j)) \hat{\gamma}(i,j)}{\sqrt{\sigma^2_\theta \hat{\beta}(i)^2 + \hat{\gamma}(i,i)} \sqrt{\sigma^2_\theta \hat{\beta}(j)^2 + \hat{\gamma}(j,j)}})^2
= \frac{\qty(\sigma^2_\theta \hat{\beta}(i) \hat{\beta}(j) + \hat{\gamma}(i,j))^2}{\sigma^2_\theta \hat{\beta}(i)^2 + \hat{\gamma}(i,i)}.
\end{align*}
Substituting this lower bound into the preceding inequality yields \eqref{id_var2}.
\end{proof}


\begin{proof}[Proof of Proposition~\ref{prop_var_eq}]
Let $\calX$ be an induced outcome of $(\calG,\bar{\varphi})$.
The inequality \eqref{id_var1} holds with equality if and only if conditioning on $S(i)$ and conditioning on $X(i) = \varphi_0(i) +\varphi^\top(i)K(i,i)^{-\frac{1}{2}}(S(i) - m(i))$ reduce the same amount of state variance.
The proof of Lemma~\ref{lem_var} reveals that this is the case if and only if
\begin{equation} \label{pf_var_nogap}
\Var\qty[ \E\qty[\theta \mid S(i)] \mid X(i)] = 0.
\end{equation}
To evaluate this condition, we calculate the relevant covariances.
First, Lemma~\ref{lem_affine} directly implies $\Var\qty[X(i)] = \varphi(i)^\top \varphi(i)$.
Second, using the conditional Gaussian formula $\E\qty[\theta \mid S(i)] = \mu_\theta + K_\theta(i)^\top K(i,i)^{-1}(S(i)-m(i))$, the covariance between the equilibrium action and the state estimate given the signal is
\[
\Cov\qty[X(i), \E\qty[\theta \mid S(i)]] = \varphi(i)^\top K(i,i)^{-\frac{1}{2}} \Var\qty[S(i)] K(i,i)^{-1} K_\theta(i)
= \sigma_\theta \varphi(i)^\top P_\theta(i).
\]
Third, the variance of the state estimate given the signal is
\[
\Var\qty[\E\qty[\theta \mid S(i)]] = K_\theta(i)^\top K(i,i)^{-\frac{1}{2}} K_\theta(i) = \sigma^2_\theta P_\theta(i)^\top P_\theta(i).
\]
Putting these together, again by the conditional Gaussian formula, we obtain
\begin{align*}
&\Var\qty[ \E\qty[\theta \mid S(i)] \mid X(i)] \\
&= \Var\qty[\E\qty[\theta \mid S(i)]] - \frac{\Cov\qty[X(i), \E\qty[\theta \mid S(i)]]^2}{\Var\qty[X(i)]} = \sigma^2_\theta \qty(P_\theta(i)^\top P_\theta(i) - \frac{\qty(\varphi(i)^\top P_\theta(i))^2}{\varphi(i)^\top \varphi(i)}).
\end{align*}
This implies that \eqref{pf_var_nogap} is satisfied if and only if $\|P_\theta(i)\| \cdot \|\varphi(i)\| = |\varphi(i)^\top P_\theta(i)|$, which is precisely the equality version of the Cauchy--Schwarz inequality.
Therefore, equality holds in \eqref{id_var1} if and only if $\varphi(i)$ and $P_\theta(i)$ are linearly dependent.

Moreover, by \eqref{int_bne}, the equilibrium slope $\varphi(i)$ satisfies
\[
\varphi(i) = \int_0^1 w(i,j)P(i,j) \varphi(j) \d j + \sigma_\theta b(i) P_\theta(i).
\]
It is now clear that if either $w(i,j)$ or $P(i,j)$ is zero for almost every $j \in [0,1]$, then $\varphi(i)$ and $P_\theta(i)$ are linearly dependent.
Alternatively, if either $b(j)$ or $P_\theta(j)$ is zero for almost every $j \in [0,1]$, then \eqref{int_bne} is satisfied with $\varphi(j)=0$ for all such $j$ and $\varphi(j') = \sigma_\theta b(j') P_\theta(j')$ for the rest of agents.\footnote{Note that $b(i)$ and $P_\theta(i)$ are nonzero in this case since $\Var[X(i)] > 0$ is assumed in Proposition~\ref{prop_var_eq}.}
Consequently, the required linear dependence again holds for $i$.
\end{proof}


\begin{proof}[Proof of Proposition~\ref{prop_actvar}]
In an LQG game exhibiting payoff-symmetry, i.e., $w(i,j) \equiv a < 1$, the slope $\varphi$ of an affine equilibrium is characterized as
\[
\varphi(i) = a \int_0^1 P(i,j) \varphi(j) \d j + \sigma_\theta b(i) P_\theta(i), \quad \forall i \in [0,1].
\]
By viewing $P$ as a linear operator as in Lemma~\ref{lem_P} and denoting $q(i) = \sigma_\theta b(i) P_\theta(i)$, this equation is rewritten as $\varphi - a P\varphi = q$.
Taking the norm of each side of this equation yields
\begin{equation} \label{pf_q_phi}
\|q\|^2 = \|\varphi-aP\varphi\|^2 = \|\varphi\|^2 - 2a\langle \varphi,P\varphi\rangle + a^2\|P\varphi\|^2.
\end{equation}
If $a=0$, then $\|q\|^2 = \|\varphi\|^2$, which verifies \eqref{eq_actvar1} with equality.

Suppose that $0<a<1$.
Since Lemma~\ref{lem_P} implies $0 \le \|P\varphi\|^2\le \langle \varphi,P\varphi\rangle$, \eqref{pf_q_phi} yields
\[
\|q\|^2 \le \|\varphi\|^2-a(2-a)\langle \varphi,P\varphi\rangle \le \|\varphi\|^2,
\]
which verifies \eqref{eq_actvar1}.
Moreover, if $Pq=0$, then applying $P$ to $\varphi-aP\varphi=q$ gives $P\varphi=aP^2\varphi$, which in turn implies $\|P\varphi\| = a \|P^2 \varphi\|$.
Moreover, since Lemma~\ref{lem_P} implies $\|P^2\varphi\| \le \|P\varphi\|$ by applying \eqref{eq_normP} for $P\varphi$, it follows that $\|P\varphi\| \le a\|P\varphi\|$, and hence $P\varphi=0$ by $a \in (0,1)$.
It follows that $\varphi=q$, so $\|q\|^2=\|\varphi\|^2$.
Thus, \eqref{eq_actvar1} holds with equality if $Pq=0$.

Suppose that $a<0$.
Since Lemma~\ref{lem_P} implies that $\langle \varphi,P\varphi\rangle\le \|\varphi\|^2$ and $\|P\varphi\|^2\le \|\varphi\|^2$, \eqref{pf_q_phi} yields
\[
\|q\|^2 \le (1-2a+a^2)\|\varphi\|^2 = (1-a)^2\|\varphi\|^2,
\]
which verifies \eqref{eq_actvar2}.
Moreover, if $Pq=q$, then applying $P$ to $\varphi-aP\varphi=q$ gives $P\varphi-aP^2\varphi=q$.
Subtracting this equation from $\varphi-aP\varphi=q$ gives $(I-aP)(I-P)\varphi=0$.
Since $a<0$ and $P$ is positive semidefinite, all eigenvalues of $I-aP$ are strictly greater than $1$.
In particular, this implies that $(I-aP)$ is invertible, and hence $(I-P)\varphi = 0$, or $P\varphi=\varphi$.
It follows that $q=(1-a)\varphi$, so $\|q\|^2=(1-a)^2\|\varphi\|^2$.
Thus, \eqref{eq_actvar2} holds with equality if $Pq=q$.
\end{proof}


\subsection{Omitted Mathematical Details}
\label{app_misc}

Appendix~\ref{app_misc} completes the proof of Theorem~\ref{thm_well} by verifying Lemmas~\ref{lem_compact} and~\ref{lem_psi_conv}.
The key mathematical result is Lemma~\ref{lem_FK}, which offers a condition for square-integrable functions to satisfy a version of equicontinuity.
This condition allows us to apply the Fr\'echet--Riesz--Kolmogorov theorem, which is also known as the $L_p$ analogue of the Arzel\`a--Ascoli theorem.

To this end, it is useful to extend the domains of $w$ and $P$ from $[0,1]^2$ to $\R^2$ by setting them equal to zero outside $[0,1]^2$.
For any matrix-valued function $M \in \calL_2^{n \times n}(\R^2)$, define the linear operator
\[
T_M \psi \coloneqq \int_\R M(\cdot,j) \psi(j) \d j, \quad \forall \psi \in \calL_2^n(\R).
\]
By mimicking the proof of Lemma~\ref{lem_HS_norm}, it is straightforward to show that is a bounded operator with
\begin{equation} \label{eq_M_norm}
\|T_M\psi\| \le \|M\|\cdot \|\psi\|, \quad \forall \psi \in \calL^n_2(\R).
\end{equation}
For $M \in \calL_2^{n \times n}(\R^2)$ and $\delta > 0$, define $M^\delta \in \calL_2^{n \times n}(\R^2)$ by
\[
M^\delta (i,j) \coloneqq M(i+\delta,j), \quad \forall i,j \in \R.
\]
Clearly, $\|M\| = \|M^\delta\|$ and $\vertt{T_M} = \vertt{T_{M^\delta}}$.

\begin{lemma} \label{lem_M_cts}
For any continuous function $M:\R^2 \to \R^{n \times n}$ supported on a compact subset of $\R^2$, $M^\delta$ converges to $M$ in $L_2$ as $\delta \to 0$.
\end{lemma}

\begin{proof}
Fix any $\bar{\delta} > 0$.
Since $M$ is supported on a compact subset of $\R^2$, we can take a compact set $A \subseteq \R^2$ that contains the support of $M^\delta$ for all $\delta \in [0,\bar{\delta})$.
Moreover, since $M$ is continuous and compactly supported, the family $\{M^\delta\}_{\delta \in [0,\bar{\delta})}$ is uniformly bounded.
Thus, there exists a constant $c>0$ such that $\vertd{M^\delta(i,j)} \le c$ for all $(i,j) \in \R^2$ and all $\delta \in [0,\bar{\delta})$.
Then, the mapping $(i,j) \mapsto \|M(i+\delta,j)-M(i,j)\|^2$ is uniformly bounded by the integrable function $4c^2 \times \1_A(\cdot)$.
Moreover, this mapping converges pointwise to $0$ as $\delta \to 0$ by the continuity of $M$.
Hence, the dominated convergence theorem yields
\[
\lim_{\delta \to 0} \int_\R \int_\R \|M(i+\delta,j)-M(i,j)\|^2 \d j \d i = 0,
\]
from which we obtain $\|M-M^\delta\| \to 0$ as $\delta \to 0$.
\end{proof}

\begin{lemma} \label{lem_FK}
Let $\Psi$ be a bounded subset of $\calL_2^n(\R)$, and let $\{M_n\}_{n \in \N}$ be a convergent sequence in $\calL_2^{n \times n}(\R^2)$ with limit $M \in \calL_2^{n \times n}(\R^2)$.
Then, for any $\epsilon > 0$, there exists $\bar{\delta} > 0$ such that
\[
\sup_{\psi \in \Psi,\, n \in \N,\, \delta \in (0,\bar{\delta})}
\left\|(T_{M_n} \psi)(\cdot + \delta) - (T_{M_n} \psi)(\cdot) \right\| < \epsilon.
\]
Moreover, if the supports of all $M_n$ are contained in a common compact set, then the collection $\{T_{M_n}\psi\}_{\psi \in \Psi,\, n \in \N}$ is pre-compact in $\calL_2^n(\R)$.
\end{lemma}

\begin{proof}
By definition,
\[
(T_{M_n} \psi)(\cdot + \delta) - (T_{M_n} \psi)(\cdot)
= \int_\R \qty(M_n^\delta(\cdot,j) - M_n(\cdot, j)) \psi (j) \d j.
\]
Hence, by applying \eqref{eq_M_norm} to $M_n^\delta - M_n$, we obtain
\[
\left\|(T_{M_n} \psi)(\cdot + \delta) - (T_{M_n} \psi)(\cdot) \right\|
\le \|M^\delta_n - M_n\| \cdot \|\psi\|.
\]
Since $\Psi$ is bounded, it suffices to show that $\|M^\delta_n - M_n\|$ can be made arbitrarily small by taking $\delta$ small, uniformly over all $n \in \N$.

Fix any $\epsilon > 0$.
Since $\|M_n - M\| \to 0$, we can take $n_0 \in \N$ such that $\|M-M_n\| < \epsilon/9$ for all $n > n_0$.
By definition, for any $n > n_0$, we also have $\|M_n^\delta - M^\delta\| < \epsilon/9$ for all $\delta > 0$.
In addition, by Theorem~4.12 of \cite{brezis}, there exists a compactly supported continuous function $\tilde{M}:\R^2 \to \R^{n \times n}$ such that $\|M-\tilde{M}\| < 2\epsilon/9$.
Then $\|M^\delta-\tilde{M}^\delta\| < 2\epsilon/9$ holds as well.
Lastly, by Lemma~\ref{lem_M_cts}, there exists $\delta_0 > 0$ such that $\|\tilde{M}-\tilde{M}^\delta\| < \epsilon/3$ for all $\delta \in (0,\delta_0)$.
Putting these estimates together, for any $n > n_0$ and $\delta \in (0,\delta_0)$, we have
\begin{align*}
\|M^\delta_n - M_n\|
\le\underbrace{\|M^\delta_n - M^\delta\|}_{<\epsilon/9}
+ \underbrace{\|M^\delta - \tilde{M}^\delta\|}_{<2\epsilon/9}
+ \underbrace{\|\tilde{M}^\delta - \tilde{M}\|}_{<\epsilon/3}
+ \underbrace{\|\tilde{M} - M\|}_{<2\epsilon/9}
+ \underbrace{\|M-M_n\|}_{<\epsilon/9}
< \epsilon.
\end{align*}
Notice that $\delta_0$ is chosen independently of $n$.

Next, consider any $n \le n_0$.
As before, let $\tilde{M}_n: \R^2 \to \R^{n \times n}$ be a compactly supported continuous function such that $\|M_n - \tilde{M}_n\| < \epsilon/4$.
Then $\|M_n^\delta - \tilde{M}_n^\delta\| < \epsilon/4$ holds as well.
Also, Lemma~\ref{lem_M_cts} implies that there exists $\delta_n > 0$ such that $\|\tilde{M}_n - \tilde{M}^\delta_n\| < \epsilon/2$ for all $\delta \in (0,\delta_n)$.
Then, for any $\delta \in (0,\delta_n)$, it follows that
\begin{align*}
\|M^\delta_n - M_n\|
\le\underbrace{\|M^\delta_n - \tilde{M}_n^\delta\|}_{<\epsilon/4}
+ \underbrace{\|\tilde{M}^\delta_n - \tilde{M}_n\|}_{<\epsilon/2}
+ \underbrace{\|\tilde{M}_n - M_n\|}_{<\epsilon/4} < \epsilon.
\end{align*}
Since such $\delta_n$ need be chosen only for finitely many $n \le n_0$, letting $\bar{\delta} = \min \{\delta_0,\delta_1,\ldots,\delta_{n_0}\} > 0$ gives the desired uniform bound for $\|M_n^\delta-M_n\|$ that holds across all $n \in \N$ and $\delta \in (0,\bar{\delta})$.

Having established the uniform boundedness, the asserted pre-compactness is a direct consequence of the Fr\'echet--Riesz--Kolmogorov theorem \citep[Theorem~4.26 of][]{brezis}.
\end{proof}

Now, the previous lemmas in Appendix \ref{app_pf} follow from Lemma \ref{lem_FK}.

\begin{proof}[Proof of Lemma \ref{lem_compact}]
Take any bounded sequence $\{\bar{\psi}^n\}_{n \in \N}$ in $\calL_2^{\bar{d}+1}[0,1]$, and let $M=M_n=wP$ for all $n \in \N$; then Lemma \ref{lem_FK} implies $\{T\bar{\psi}^n\}_{n \in \N}$ is pre-compact in $\calL_2^{\bar{d}+1}[0,1]$, which shows that $T$ is a compact linear operator.
\end{proof}

\begin{proof}[Proof of Lemma \ref{lem_psi_conv}]
Let $\{\bar{\psi}^n\}_{n \in \N}$ be given as in Lemma \ref{lem_psi_conv}, i.e., $\|\bar{\psi}^n\| = 1$ and $\bar{\psi}^n = T_n \bar{\psi}^n$ for all $n \in \N$.
Trivially, $\{\bar{\psi}^n\}_{n \in \N}$ is bounded.
In addition, as shown in the proof of Lemma \ref{lem_cts_op}, it holds that $\|wP-w_nP_n\| \to 0$ as $n \to \infty$ (note that the verification of this claim depends neither on Lemma \ref{lem_compact} nor \ref{lem_psi_conv}).
Hence, Lemma \ref{lem_FK} implies $\{T\bar{\psi}^n\}_{n \in \N} = \{\bar{\psi}^n\}_{n \in \N}$ is pre-compact in $\calL_2^{\bar{d}+1}[0,1]$, and thus, a convergent subsequence can be taken from it.
\end{proof}


\newpage
\bibliography{reference}

@article{smolinyamashita2026,
	author = {Alex Smolin and Takuro Yamashita},
	journal = {Working Paper},
	title = {{Information Design in Smooth Games}},
	year = {2026}}

@article{sun2006,
	author = {Sun, Yeneng},
	journal = {Journal of Economic Theory},
	number = {1},
	pages = {31--69},
	publisher = {Elsevier},
	title = {{The Exact Law of Large Numbers via Fubini Extension and Characterization of Insurable Risks}},
	volume = {126},
	year = {2006}}

@article{miyashita_ui_large,
	author = {Masaki Miyashita and Takashi Ui},
	journal = {Working Paper},
	title = {{On the Pettis Integral Approach to Large Population Games}},
	year = {2024}}

@book{vives1999,
	author = {Vives, Xavier},
	publisher = {MIT press},
	title = {{Oligopoly Pricing:\ Old Ideas and New Tools}},
	year = {1999}}

@article{karni2009,
	author = {Karni, Edi},
	journal = {Econometrica},
	number = {2},
	pages = {603--606},
	title = {{A Mechanism for Eliciting Probabilities}},
	volume = {77},
	year = {2009}}

@article{bhm2021,
	author = {Bergemann, Dirk and Heumann, Tibor and Morris, Stephen},
	journal = {The RAND Journal of Economics},
	number = {1},
	pages = {125--150},
	title = {{Information, Market Power, and Price Volatility}},
	volume = {52},
	year = {2021}}

@article{ui2016,
	author = {Ui, Takashi},
	journal = {Journal of Mathematical Economics},
	pages = {139--146},
	title = {{Bayesian Nash Equilibrium and Variational Inequalities}},
	volume = {63},
	year = {2016}}

@article{blume2015,
	author = {Blume, Lawrence E and Brock, William A and Durlauf, Steven N and Jayaraman, Rajshri},
	journal = {Journal of Political Economy},
	number = {2},
	pages = {444--496},
	title = {{Linear Social Interactions Models}},
	volume = {123},
	year = {2015}}

@article{bhm2015,
	author = {Bergemann, Dirk and Heumann, Tibor and Morris, Stephen},
	journal = {Journal of Economic Theory},
	pages = {427--465},
	title = {{Information and Volatility}},
	volume = {158},
	year = {2015}}

@article{lu2016,
	author = {Lu, Jay},
	journal = {Econometrica},
	number = {6},
	pages = {1983--2027},
	title = {{Random Choice and Private Information}},
	volume = {84},
	year = {2016}}

@article{arieli2017,
	author = {Arieli, Itai and Mueller-Frank, Manuel},
	journal = {Games and Economic Behavior},
	pages = {455--461},
	title = {{Inferring Beliefs from Actions}},
	volume = {102},
	year = {2017}}

@article{libgober2025,
	author = {Libgober, Jonathan},
	journal = {Journal of Political Economy Microeconomics},
	number = {4},
	pages = {798--826},
	title = {{Identifying wisdom (of the crowd): A regression approach}},
	volume = {3},
	year = {2025}}

@article{savage1971,
	author = {Savage, Leonard J},
	journal = {Journal of the American Statistical Association},
	number = {336},
	pages = {783--801},
	title = {{Elicitation of Personal Probabilities and Expectations}},
	volume = {66},
	year = {1971}}

@article{lu2019,
	author = {Lu, Jay},
	journal = {American Economic Review},
	number = {9},
	pages = {3192--3228},
	title = {{Bayesian Identification:\ A Theory for State-dependent Utilities}},
	volume = {109},
	year = {2019}}

@article{blackwell1951,
	author = {Blackwell, David},
	journal = {Proceedings of the Second Berkeley Symposium on Mathematical Statistics and Probability},
	pages = {93--102},
	title = {{Comparison of Experiments}},
	volume = {1},
	year = {1951}}

@article{hansen1974,
	author = {Hansen, Ole Havard and Torgersen, Erik N},
	journal = {The Annals of Statistics},
	number = {2},
	pages = {367--373},
	title = {{Comparison of Linear Normal Experiments}},
	volume = {2},
	year = {1974}}

@article{pettis1938,
	author = {Pettis, Billy James},
	journal = {Transactions of the American Mathematical Society},
	number = {2},
	pages = {277--304},
	title = {{On Integration in Vector Spaces}},
	volume = {44},
	year = {1938}}

@article{manski1993,
	author = {Manski, Charles F},
	journal = {The Review of Economic Studies},
	number = {3},
	pages = {531--542},
	title = {{Identification of Endogenous Social Effects:\ The Reflection Problem}},
	volume = {60},
	year = {1993}}

@article{lee2007,
	author = {Lee, Lung-Fei},
	journal = {Journal of Econometrics},
	number = {2},
	pages = {333--374},
	title = {{Identification and Estimation of Econometric Models with Group Interactions, Contextual Factors and Fixed Effects}},
	volume = {140},
	year = {2007}}

@article{bramoulle2009,
	author = {Bramoull{\'e}, Yann and Djebbari, Habiba and Fortin, Bernard},
	journal = {Journal of Econometrics},
	number = {1},
	pages = {41--55},
	title = {{Identification of Peer Effects through Social Networks}},
	volume = {150},
	year = {2009}}

@article{vives1984,
	author = {Vives, Xavier},
	journal = {Journal of Economic Theory},
	number = {1},
	pages = {71--94},
	title = {{Duopoly Information Equilibrium: Cournot and Bertrand}},
	volume = {34},
	year = {1984}}

@article{uhlig1996,
	author = {Harald Uhlig},
	journal = {Economic Theory},
	pages = {41--50},
	title = {{A Law of Large Numbers for Large Economies}},
	volume = {8},
	year = {1996}}

@article{judd1985,
	author = {Judd, Kenneth L.},
	journal = {Journal of Economic Theory},
	number = {1},
	pages = {19--25},
	title = {{The Law of Large Numbers with a Continuum of IID Random Variables}},
	volume = {35},
	year = {1985}}

@article{lop2018,
	author = {Lambert, Nicolas S and Ostrovsky, Michael and Panov, Mikhail},
	journal = {Econometrica},
	number = {4},
	pages = {1119--1157},
	title = {{Strategic Trading in Informationally Complex Environments}},
	volume = {86},
	year = {2018}}

@article{alnajjar1995,
	author = {Nabil Ibraheem Al-Najjar},
	journal = {Econometrica},
	number = {5},
	pages = {1195--1224},
	title = {{Decomposition and Characterization of Risk with a Continuum of Random Variables}},
	volume = {63},
	year = {1995}}

@article{bm2016,
	author = {Dirk Bergemann and Stephen Morris},
	journal = {Theoretical Economics},
	number = {2},
	pages = {487--522},
	title = {{Bayes Correlated Equilibrium and the Comparison of Information Structures in Games}},
	volume = {11},
	year = {2016}}

@article{bramoulle2014,
	author = {Bramoull{\'e}, Yann and Kranton, Rachel and D'amours, Martin},
	journal = {American Economic Review},
	number = {3},
	pages = {898--930},
	title = {{Strategic Interaction and Networks}},
	volume = {104},
	year = {2014}}

@article{bhm2017,
	author = {Bergemann, Dirk and Heumann, Tibor and Morris, Stephen},
	journal = {Working Paper},
	title = {{Information and Interaction}},
	year = {2017}}

@article{lmo2018,
	author = {Lambert, Nicolas S and Martini, Giorgio and Ostrovsky, Michael},
	journal = {Working Paper},
	title = {{Quadratic Games}},
	year = {2018}}

@article{miyashita_ui_lqgd,
	author = {Masaki Miyashita and Takashi Ui},
	journal = {Working Paper},
	title = {{LQG Information Design}},
	year = {2025}}

@book{horn,
	author = {Horn, Roger A and Johnson, Charles R},
	edition = {2},
	publisher = {Cambridge University Press},
	title = {Matrix Analysis},
	year = {2012}}

@article{graphon2023,
	author = {Parise, Francesca and Ozdaglar, Asuman},
	journal = {Econometrica},
	number = {1},
	pages = {191--225},
	title = {{Graphon Games:\ A Statistical Framework for Network Games and Interventions}},
	volume = {91},
	year = {2023}}

@book{brezis,
	author = {Brezis, Haim},
	publisher = {Springer Science \& Business Media},
	title = {Functional Analysis, Sobolev Spaces and Partial Differential Equations},
	year = {2010}}

@book{lax,
	author = {Peter D. Lax},
	publisher = {John Wiley \& Sons},
	title = {Functional Analysis},
	year = {2002}}

@book{krees,
	author = {Rainer Kress},
	edition = {3},
	publisher = {Springer, New York},
	title = {Linear Integral Equations},
	year = {2014}}

@book{dudley,
	author = {Richard M. Dudley},
	publisher = {Wadsworth, Belmont, CA},
	title = {Real Analysis and Probability},
	year = {1989}}

@book{ab2006,
	author = {Aliprantis, Charalambos D and Border, C. Kim},
	publisher = {Springer},
	title = {Infinite Dimensional Analysis},
	year = {2006}}

@article{ap2007,
	author = {Angeletos, George-Marios and Pavan, Alessandro},
	journal = {Econometrica},
	number = {4},
	pages = {1103--1142},
	publisher = {Wiley Online Library},
	title = {{Efficient Use of Information and Social Value of Information}},
	volume = {75},
	year = {2007}}

@article{bm2013,
	author = {Bergemann, Dirk and Morris, Stephen},
	journal = {Econometrica},
	number = {4},
	pages = {1251--1308},
	title = {{Robust Predictions in Games with Incomplete Information}},
	volume = {81},
	year = {2013}}

@article{ms2002,
	author = {Morris, Stephen and Shin, Hyun Song},
	journal = {American Economic Review},
	number = {5},
	pages = {1521--1534},
	title = {{Social Value of Public Information}},
	volume = {92},
	year = {2002}}

@article{ui2015,
	author = {Ui, Takashi and Yoshizawa, Yasunori},
	journal = {Journal of Economic Theory},
	pages = {507--535},
	title = {{Characterizing Social Value of Information}},
	volume = {158},
	year = {2015}}

@article{ui2013,
	author = {Ui, Takashi and Yoshizawa, Yasunori},
	journal = {Economics Bulletin},
	number = {1},
	pages = {72--77},
	title = {{Radner's Theorem on Teams and Games with a Continuum of Players}},
	volume = {33},
	year = {2013}}

@article{radner1962,
	author = {Radner, Roy},
	journal = {The Annals of Mathematical Statistics},
	number = {3},
	pages = {857--881},
	title = {{Team Decision Problems}},
	volume = {33},
	year = {1962}}

@book{manski1995,
	author = {Manski, Charles F},
	publisher = {Harvard University Press},
	title = {{Identification Problems in the Social Sciences}},
	year = {1995}}


\end{document}